\documentclass[11pt]{article}
\usepackage[utf8]{inputenc}
\usepackage{csquotes}
\usepackage{setspace}
\usepackage[dvipsnames]{xcolor}
\usepackage[margin=1in]{geometry}
\usepackage{hyperref}
\usepackage{amsmath, amssymb, amscd, amsthm, amsfonts}
\usepackage{graphicx}
\usepackage{adjustbox}
\usepackage{hyperref}
\usepackage{cleveref}
\usepackage{amsthm}
\usepackage[utf8]{inputenc}
\usepackage[english]{babel}
\usepackage{proof}
\usepackage{ninecolors}

\usepackage{proof}
\usepackage{tikz-cd}
\usepackage{xcolor}
\usepackage{thmtools, thm-restate}
\tikzcdset{scale cd/.style={every label/.append style={scale=#1},
		cells={nodes={scale=#1}}}}

\newtheorem{theorem}{Theorem}[section]

\newtheorem{lemma}[theorem]{Lemma}

\theoremstyle{definition}
\newtheorem{definition}{Definition}[section]

\title{Simplicial Semantics for Belief Revision}
\author{Philip Sink}
\date{}

\newcommand{\M}{\mathcal{M}}

\newcommand{\commentout}[1]{}
\newcommand{\defin}[1]{\textbf{#1}}

\renewcommand{\phi}{\varphi}

\newcommand{\lthen}{\rightarrow}

\newcommand{\F}{\mathcal{F}}
\renewcommand{\M}{\mathcal{M}}
\renewcommand{\L}{\mathcal{L}}

\usepackage{amssymb,tikz}

\begin{document}
	
	\maketitle
	
\begin{abstract}
	This paper will give a definition of belief revision within simplicial semantics. Starting from the work presented in \cite{BelSimp} as a baseline, this paper modifies the semantics for belief given there to allow atomic formulae to be assigned to nodes, not facets. Conceptually and philosophically, this version of the semantics is better suited if one wishes to interpret the nodes of a simplicial model of epistemic logic as a ``perspective'' assigned to a particular agent. If nodes are perspectives, it then follows that worlds, or the facets of a simplicial model, are composed themselves of perspectives. This allows us to say that two worlds are more similar, or ``nearer'', if they share more perspectives. With this conceptual notion of nearness in hand, two different formal presentations of revision are given. We conclude by exploring some conceptual pitfalls surrounding these definitions under iterated revision, and motivate a few potential solutions that involve giving the agents a ``memory'' of what has been announced so far.
\end{abstract}

\section{Introduction} \label{2sec:int}

Semantics for epistemic logic in simplicial complexes has been an ongoing area of research in the last few years.  \cite{Death,DoA,KaSC,SimpDEL,HG,FA,SimpSet,SimpBel,BelSimp} Some of the earliest work in this area explored connections between simplicial semantics and Dynamic Epistemic Logic, particularly public announcements. \cite{SimpDEL} Recent work has attempted to use this setting to model belief. \cite{SimpBel,BelSimp,KaSC} What is relevant about belief, as compared to knowledge, is that belief is not generally taken to be factive. That is, the axiom schema \textbf{T}, often called factivity, is not sound. As a result, agents are able to have false beliefs. It may be the case that for some world $w$, and some formula $\varphi$, we have that $B_a\varphi$ is true (read ``agent $a$ believes that $\varphi$ is true''), but $\varphi$ is false.

This paper will explore the overlap between false beliefs and public announcements in simplicial semantics. In general, if $\varphi$ is publicly announced, then any world $w$ which agent $a$ considers possible, they cease to consider possible. So, if agent $a$ believes $\varphi$, and then $\neg\varphi$ is announced, then agent $a$ no longer considers any worlds possible. Their beliefs have become defunct. Belief revision is an alternative way of modeling how agents learn information in an announcement which avoids this issue. \cite{Lewis2,Grove,AGM} The main idea in revision is that some worlds are ``near'' to each other. If agent $a$ considers a world $w$ possible which satisfies $\varphi$, and $\neg\varphi$ is announced, the agent may ``replace'' or ``revise'' $w$ to a ``nearby'' world $w'$ that satisfies $\neg\varphi$.

The main issue with this notion of nearness is that it is arbitrary, in practice, to determine what worlds should be nearer or further from each other. We believe that simplicial semantics offers a solution here. In simplicial semantics, it is often assumed that worlds are composed of a simpler epistemic structure, namely ``perspectives''. \cite{HG,KaSC} If this is true, we can say that two worlds are nearer to each other if they share more perspectives. this paper defines a notion of revision in simplicial semantics that operates on precisely this notion. In Section \ref{2sec:persp}, we outline the epistemology of ``perspectives'' that we will appeal to in this paper. In Section \ref{2sec:bel}, we briefly outline the semantics for belief that this paper will appeal to, which is a modified form of the semantics in \cite{BelSimp}. In Section \ref{2sec:revinf}, we give a detailed motivation of the concept of nearness at play in this paper, and we formalize it in Section \ref{2sec:revform}. We then proceed to give some examples, motivate a potential issue that arises if we assume announcements are indefeasible, and offer two potential solutions.

\section{The ``Perspective'' Perspective}\label{2sec:persp}

The basic epistemic unit that we will start with we will call a ``perspective'' or ``point of view,'' very similar to the idea outlined in \cite{HG}. Given a particular agent, let's call her Alice, the perspective of that agent consists abstractly of (some of) the information available to that agent. For example, if the agent is sitting on a park bench, and sees a blue sky and ducks in the pond, all of this information is what we might consider part of the agent's perspective. This is indeed an abstract notion, but crucial to us is that the information constituting perspectives can be broken down into two pieces. The first is the \textit{hard} information, or the information that is in some sense indefeasible. In this scenario, where the agent on the park bench sees ducks, we consider the proposition representing this fact hard information. By contrast, suppose Alice has a friend, Barbara, who regularly goes to the park with her. Each time the pair has been to the park, they see ducks; however today, her friend is home sick. Alice still goes to the park and sees ducks, but Barbara of course does not. In this example, we will assume that both believe that there are ducks in the pond. However, for Alice, this is hard information. By contrast, Barbara believes there are ducks in the pond in a weaker sense. This is an example of \textit{weak} information.

The line between hard and weak information is only intuitive. In fact, the distinction only matters as a modeling convention. In certain modeling contexts, it will be useful to represent some information as indefeasible, and other information as defeasible. That is all the distinction between \textit{hard} and \textit{weak} is meant to cache out. It rests solely on the intuition that an agent's point of view is made up of facts that are immediately obvious to the agent, and other facts that are less so. As a technical and simplifying assumption, we will assume that the hard information is given by a possibly empty set of literals associated with that perspective.

Since we are taking perspectives as our foundational epistemic object, how can we use them to model uncertainty? One again imagines our agent at the pond with the ducks. Call her agent $a$. Her friend will be agent $b$. If $D$ is the logical atom that means ``The ducks are in the pond at the park,'' then we can say that $D$ is assigned to $a$'s perspective, call it $a_1$. $b$'s perspective does not, in this language, consist of any hard information, and so we associate $b_1$ with the empty set. Since these perspectives are true simultaneously, we can group them together as the set $\{a_1,b_1\}$, as drawn in Figure \ref{2fgr:ex1}

\begin{figure}[htbp]
	\begin{center}
		\begin{tikzcd}
			a_1(D) \arrow[r, no head] & b_1
		\end{tikzcd}
		\caption{A model with two perspectives and one world/facet.}\label{2fgr:ex1}
		
	\end{center}
\end{figure}

When we formally define simplicial semantics later, it will be clear that Figure \ref{2fgr:ex1} is a (very basic) simplicial complex. For readers unfamiliar with the exact workings of simplicial semantics, one should look at \cite{KaSC} and \cite{BelSimp}. In order to read the above, for now, it suffices to recognize that Figure \ref{2fgr:ex1} is equivalent to a Kripke model, with one world where $D$ is true, and two reflexive edges on that world, one for each agent. The single edge between the two perspectives in Figure \ref{2fgr:ex1} represents the set containing those two perspectives. This set with the two perspectives is the possible world, and each agent recognizes that world in their perspective, hence the reflexive edges.

Now let's imagine a slightly more complicated example. Both friends are going to the park, but do not see each other yet. We will consider this from $a$'s point of view. She sees the ducks in the pond, but is uncertain whether or not $b$ sees the ducks yet. She cannot imagine that $b$ sees there are not ducks in the pond, as that would be inconsistent with her own experience. But, she cannot tell whether or not $b$'s perspective is one where they are, in fact, seeing ducks, or they have not seen the ducks yet. The latter is the same as $b_1$ as above, but the former, which we call $b_2$, will have $D$ associated with it, as seen in Figure \ref{2fgr:ex2}.

\begin{figure}[htbp]
	\begin{center}
		\begin{tikzcd}
			b_2(D) \arrow[red, r, no head] &  a_1(D) \arrow[red, r, no head] & b_1
		\end{tikzcd}
		\caption{A model with three perspectives and two worlds/facets.}\label{2fgr:ex2}
		
	\end{center}
\end{figure}

Figure \ref{2fgr:ex2} is equivalent to a Kripke model with two worlds, one for each set, shown as edges. The worlds are indistinguishable to agent $a$, and all reflexive edges for both agents are present. In the case of two agents, we will refer to sets containing two perspectives as \textit{facets}. The facets in Figure \ref{2fgr:ex2} are red because they are associated with agent $a$. That is, these are facets, or worlds, that $a$ considers possible. In general, in this paper, we will assume that $a$ is red, $b$ is blue, and $c$ is green.

\section{Simplicial Semantics for Belief} \label{2sec:bel}

This paper will work in a slightly modified version of the semantics for belief, as presented in \cite{BelSimp}. We will reproduce here the relevant formal details for those unfamiliar, as well as show how we need to modify the semantics for our purposes. The modification itself is similar. The semantics in \cite{BelSimp} functions by assigning atomic propositions directly to the facets of the simplicial complexes. This follows the example of much of the previous literature. \cite{Death,FA,SimpSet} By contrast, here we will assign atomic propositions to the nodes of the simplicial complexes. This, too, has a great deal of precedent in the literature. \cite{DoA,KaSC,SimpDEL,HG} The reason for this change is that the notion of revision which we set out to motivate only makes sense if we can think of the individual nodes of a simplicial complex as \emph{perspectives} or ``points of view'' belonging to a particular agent. That is, a node delineates what an agent is or is not ``seeing'', or is aware of, in some immediate sense. We will return to this in Section \ref{2sec:revform}.

Let $\mathfrak{P}$ be a countable set of propositional atoms, and $Ag$ a finite set of agents. Let $N$ be a set of of \emph{nodes}, $V:N\rightarrow Ag$ a function which assigns each node to an agent, called the \emph{coloring} function, and $L:N\rightarrow 3^\mathfrak{P}$ a function which assigns each node to a set of literals, which we call the \emph{assignment}. The interpretation is that $L(n)(P)=1$ if and only if $P$ is associated with $n$, $L(n)(P)=0$ if and only if $\neg P$ is associated with $n$, and $L(n)(P)=2$ if and only if neither is associated with $n$. $L$ is the only difference between their presentation here as compared to \cite{BelSimp}. As mentioned above, $L$ assigns propositions to nodes.

Our first key idea is that we can use $N$, $V$, and $L$ to create a kind of \defin{maximal} simplicial complex. Like much of the previous literature, we will assume our simplicial complexes are uniquely colored. Specifically, each facet of our complexes has a dimension of size $|Ag|$, and no two nodes are associated to the same agent. Put formally, if $X\in S$ is such that for all $Y\in S$, if $X\subseteq Y$ then $X=Y$, we have that $|X|=|Ag|$, and, if $x,y\in X$, then if $V(x)=V(y)$, we have that $x=y$. We call this condition UCF for ``Uniquely Colored Facets''. The \defin{Maximal Complex} is the set of subsets of $N$ such that the associated logical content with that subset is consistent, and it satisfies the UCF condition. That is, it's the subsets $x\subseteq N$ such that $|x|=|A|$, $\bigcup_{u\in x}L(u)$ is consistent, and for all $u,v\in x$, $V(u)\neq V(v)$. We refer to $\mathfrak{M}(N,V,L)$ as the \defin{Maximal Complex} of $N$, $V$, and $L$. When the context is clear, we will refer to it simply as the maximal complex. The maximal complex can be defined set theoretically as follows:

\begin{align*}
	\mathfrak{M}(N,V,L):=\{y\in 2^N~|~\exists x\in 2^N (&y\subseteq x\\&\wedge\neg(\exists P\in\mathfrak{P}(\exists u,v\in x(L(u)(P)=1\wedge L(v)(P)=0))) \\&\wedge (|x|=|A|)
	\\&\wedge (\forall u,v\in x(V(u)\neq V(v))))\}
\end{align*}

Note that $\mathfrak{M}(N,V,L)$ is a simplicial complex. This is because, if $Y\in \mathfrak{M}(N,V,L)$, there is a set $X$ such that $Y\subseteq X$ and $X$ satisfies certain properties. If $Z\subseteq Y$, then the same set $X$ suffices to show that $Z\in \mathfrak{M}(N,V,L)$. All complexes we consider in this paper will be UCF subcomplexes of the maximal complex. 

If $S$ is a simplicial complex, let $\mathcal{F}(S)$ denote the facets\footnote{Maximal faces under the subset ordering.} of $S$. Since all simplicial complexes we consider are UCF, we can define projection functions $\pi_a:\mathcal{F}(S)\rightarrow N$ given by $\pi(a)(X)=V^{-1}(a)\cap X$. That is, $\pi_a$ maps each facet to the unique $a$-colored node it contains.

We say a \defin{simplicial belief model} is a tuple $(N,V,L,S,\{S_a\}_{a\in Ag})$, where $N$, $V$, and $L$ are as above, $S$ is a UCE subcomplex of $\mathfrak{M}(N,V,L)$, and each $S_a$ is a UCF subcomplex of $S$. We will refer to each $S_a$ as the \emph{belief subcomplex} belonging to agent $a$.\footnote{The addition of multiple subcomplexes, one for each agent, is a modification to the standard simplicial semantics made to adapt those semantics for modeling belief. For more detail, see \cite{BelSimp}.} Often, we will leave $S$ unspecified. In this context, we assume that $S=\mathfrak{M}(N,V,L)$.\footnote{It is possible to interpret a knowledge modality $K_a$ on $S$ in the standard way such modalities are given in the literature. \cite{KaSC} This would have the property that knowledge implies belief, i.e., for all $a\in Ag$, $K_a\varphi\rightarrow B_a\varphi$, as in \cite{BelSimp}. We do not discuss this here in order to keep the focus on revision.} 

In this setting, the language $\L_{B}(Ag)$ recursively defined by
$$\varphi ::= P \, | \, \bot \, | \, \phi \lthen \psi \, | \, B_{a}\phi,$$
where $P \in \mathfrak{P}$ and $a \in Ag$, can be interpreted in simplicial models as follows for any $X\in\F(S)$:
\begin{align*}
	\M,X & \models P \text{ iff } \exists a\in Ag(L(\pi_a(X))(P)=1)\\
	\M,X & \nvDash \bot\\
	\M,X & \models \phi \lthen \psi \text{ iff } \M,X \models \phi \text{ implies } \M,X \models \psi\\
	\M,X & \models B_a \phi \text{ iff } \forall Y \in \F(S_a) \text{ if } \pi_a(Y) = \pi_a(X) \text{ then } \M,Y \models \phi
\end{align*}

Similar to what is described in \cite{HG}, we think of the nodes $N$, with coloring and assignment functions $V$ and $L$, as \emph{perspectives}. A perspective, as described in Section \ref{2sec:persp}, is an informal notion meant to capture information that is ``immediately'' available to the agent. It often is useful to think of the perspective as capturing those logical atoms whose truth value the agent is directly observing. So, if $P$ and $\neg R$ are assigned to an $a$-colored node $n$, we think of $n$ as representing the perspective where $a$ is directly observing that $P$ is true and $R$ is false.

Often in the literature, atoms are assumed to be assigned to the agents themselves, and thus an atom represents a fact about the particular agent it is assigned to. Such a setting is referred to as the language of \defin{Local Variables}. \cite{KaSC} In this setting, all atomic propositions belong to or are assigned to a particular agent, say agent $a$, and $a$-colored nodes fix the truth value of all $a$-colored propositions. More specifically, for each $a\in Ag$, fix a set of atoms $\mathfrak{P}_a$. We will refer to these as the atoms \defin{local} to agent $a$. The total collection of atoms will be $\bigcup_{a\in Ag}\mathfrak{P}_a$. Fix $P\in\mathfrak{P}_a$. Then for each $n\in N$, we have that $L(n)(P)\in 2$ and is only defined if $V(n)=a$. Finally, the semantic condition is modified so that $\M,X\models P$ iff $L(\pi_a(X))(P)=1$. In this way, the $a$-colored perspective of any facet fixes the truth value of all atoms in $\mathfrak{P}_a$.

There are many contexts in which presenting simplicial semantics this way is entirely appropriate. In particular, it's clear how to model scenarios where each agent has a privately held bit value, and is attempting to signal its value to other agents. \cite{DoA,SimpDEL} However, our presentation has the advantage of being slightly more general. Atomic formulae can be assigned to any agent's perspective, and so a given proposition doesn't have its truth value ``fixed'' by a particular agent at every world. This allows the methodology shown here to model a wider range of scenarios than the language of local variables. One such case is the classic example of the ``Muddy Children Problem''. \cite{DC5} We will present the case with three agents. Each agent, or child, can observe the muddiness of the other two children. So, agent $a$ has in their perspective either $M_b$ or $\neg M_b$, and either $M_c$ or $\neg M_c$. This gives twelve total perspectives, four for each agent. If one is familiar with the usual Kripke structure for this scenario, it's easy to check that at each world, every agent considers exactly two worlds possible. If we consider agent $a$, these are the world itself, and the world which differs only in the truth value of $M_a$. So, using the usual translation between Kripke models and simplicial models,\footnote{See, e.g., \cite{BelSimp}.} we have that each facet corresponding to a world is a triangle, and any two triangles share at most one perspective. Moreover, each perspective is present on exactly two triangles. The result is drawn in Figure \ref{2fgr:MCPS}. 

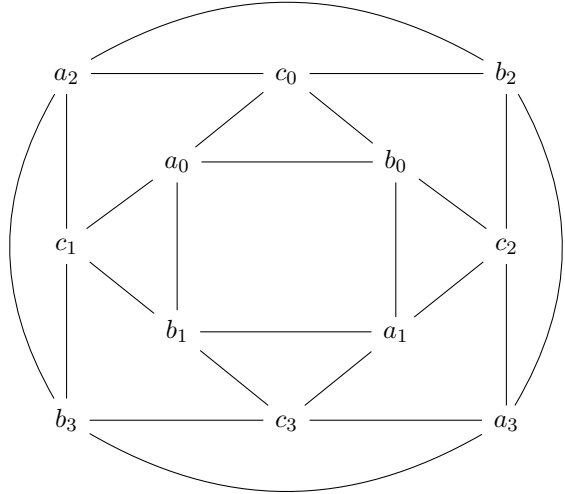
\begin{figure}[htbp]
	\begin{center}\adjustbox{scale=.9}{
			\begin{tikzcd}
				{a_2} \arrow[rr, no head] \arrow[rrrr, no head, bend left]      &                                                             & {c_0} \arrow[rd, no head] \arrow[rr, no head]           &                                                             & {b_2} \arrow[dd, no head] \arrow[dddd, no head, bend left]      \\
				& {a_0} \arrow[rr, no head] \arrow[ru, no head]      &                                                                  & {b_0} \arrow[dd, no head] \arrow[rd, no head]      &                                                                               \\
				{c_1} \arrow[ru, no head] \arrow[uu, no head]                   &                                                             &                                                                  &                                                             & {c_2} \arrow[ld, no head] \arrow[dd, no head]                   \\
				& {b_1} \arrow[uu, no head] \arrow[lu, no head] &                                                                  & {a_1} \arrow[ll, no head] \arrow[ld, no head] &                                                                               \\
				{b_3} \arrow[uu, no head] \arrow[uuuu, no head, bend left] &                                                             & {c_3} \arrow[lu, no head] \arrow[ll, no head] &                                                             & {a_3} \arrow[ll, no head] \arrow[llll, no head, bend left]
		\end{tikzcd}}
		\caption{A simplicial complex capturing the following conditions: All facets are triangles, every node is in exactly two facets, and there are twelve nodes.}\label{2fgr:MCPS}
		
	\end{center}
\end{figure}

Now we must assign atomic formulae to the individual notes. Recall that, in the Muddy Children Problem, every child can see every other child. Hence, every child can observe the muddiness of every other child. Following our intuition for hard information, this would mean that, for every $i\in Ag$, and $j,k\neq i$, both the truth value of $M_j$ and $M_k$ are \textit{hard} information for agent $i$. Hence, at each $i_n$, either $L(i_n)(M_j)=1$ or $L(i_n)(M_j)=0$, and the same for $M_k$. In particular, because $i$ is uncertain about $M_i$, $L(i_n)(M_i)=2$. We draw this in Figure \ref{2fgr:MCP}.

\begin{figure}[htbp]
	\begin{center}\adjustbox{scale=.9}{
			\begin{tikzcd}
				{a_2(M_b,\neg M_c)} \arrow[rr, no head] \arrow[rrrr, no head, bend left]      &                                                             & {c_0(M_a,M_b)} \arrow[rd, no head] \arrow[rr, no head]           &                                                             & {b_2(M_a,\neg M_c)} \arrow[dd, no head] \arrow[dddd, no head, bend left]      \\
				& {a_0(M_b,M_c)} \arrow[rr, no head] \arrow[ru, no head]      &                                                                  & {b_0(M_a,M_c)} \arrow[dd, no head] \arrow[rd, no head]      &                                                                               \\
				{c_1(\neg M_a,M_b)} \arrow[ru, no head] \arrow[uu, no head]                   &                                                             &                                                                  &                                                             & {c_2(M_a,\neg M_b)} \arrow[ld, no head] \arrow[dd, no head]                   \\
				& {b_1(\neg M_a,M_c)} \arrow[uu, no head] \arrow[lu, no head] &                                                                  & {a_1(\neg M_b,M_c)} \arrow[ll, no head] \arrow[ld, no head] &                                                                               \\
				{b_3(\neg M_b,\neg M_c)} \arrow[uu, no head] \arrow[uuuu, no head, bend left] &                                                             & {c_3(\neg M_a,\neg M_b)} \arrow[lu, no head] \arrow[ll, no head] &                                                             & {a_3(\neg M_b,\neg M_c)} \arrow[ll, no head] \arrow[llll, no head, bend left]
		\end{tikzcd}}
		\caption{A simplicial model capturing the initial configuration of the Muddy Children problem with 3 agents. Note that $S_a=S_b=S_c$. In general, when a face is present for every agent, we will draw the face as black.}\label{2fgr:MCP}
		
	\end{center}
\end{figure}
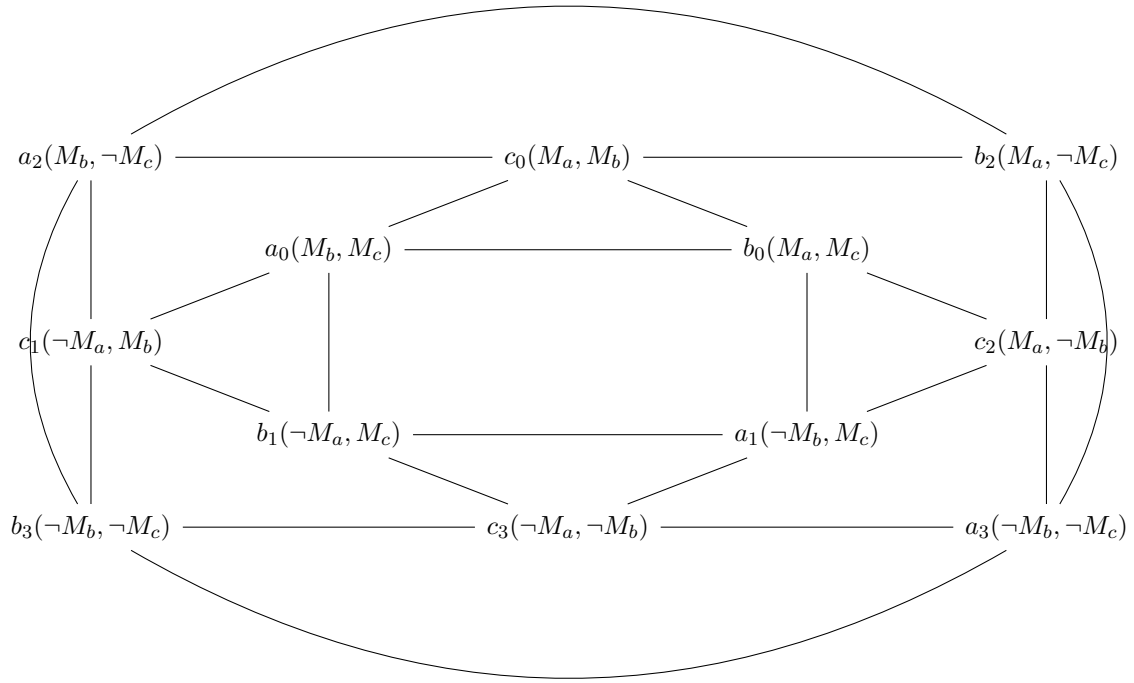

This configuration is surprising from the perspective of the language of local variables. \cite{KaSC} In this scenario, for every $M_i$, there must be an agent $j\in Ag$ such that for all $n\in\mathbb{N}$, $L(j_n)(M_i)=0$ or $L(j_n)(M_i)=1$. A natural first guess to what this might be is that $i=j$. That is, assign the proposition which stands for ``$i$ is Muddy'' to agent $i$. However, as we see in the argument above, this does not give the intended truth values for modal propositions. Indeed, such an assignment would make true $B_iM_i\vee B_i\neg M_i$ for each $i\in Ag$, which contradicts that in the initial configuration, $i$ must be \textit{uncertain} about the truth value of $M_i$. While this isn't borne out directly by the Muddy Children Problem, our system has the advantage that for a particular atom $P$ and any agent $i\in Ag$, it can be the case that there is a model $\M$ and a facet $X$ such that $\M,X\vDash \neg B_iP\vee\neg B_i\neg P$. That is, for any agent, there can be a facet where that agent is uncertain about the truth value of $P$. This cannot be the case in the language of local variables, where agents are always certain about the truth values of their local variables. This is the sense in which our system is more general than local variables.

We can spell out this greater generality more specifically. In the language of local variables, the following axiom schema, locality, is sound. Let $P\in\mathfrak{P}_a$:

$$\textbf{LOC}:P\rightarrow B_aP$$

If $P$ is local to agent $a$ and true, then $a$ believes that $P$. By contrast, in our setting, the following schema is sound for any $P\in\mathfrak{P}$:

$$\textbf{NU}:P\rightarrow\bigvee_{a\in Ag}B_aP$$

The name $\textbf{NU}$ stands for ``No Uncertainties''. The schema states that if an atom is true, then some agent believes it to be true. However, unlike in the language of local variables, which agent believes the atom to be true can vary from world to world.

It is clear that under the axioms of propositional logic, $\textbf{LOC}\vdash\textbf{NU}$. The converse is generally not true. Consider the model in Figure \ref{2fgr:NUnLOC}.

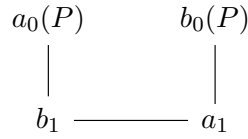
\begin{figure}[htbp]
	\begin{center}
		\begin{tikzcd}
			a_0(P) \arrow[d, no head] & b_0(P)                 \\
			b_1 \arrow[r, no head]    & a_1 \arrow[u, no head]
		\end{tikzcd}
		\caption{A simplicial belief model with two agents and three worlds/facets.}\label{2fgr:NUnLOC}
	\end{center}
\end{figure}

There are three facets, $\{a_0,b_1\}$, $\{a_1,b_0\}$, and $\{a_1,b_1\}$. Call these $w_0$, $w_1$, and $w_2$ respectively. Facets $w_0$ and $w_1$ are the only two facets where $P$ is true. However, at $w_0$, $B_a P$ is true but $B_b P$ is false. Conversely, at $w_1$, $B_aP$ is false and $B_bP$ is true. So, \textbf{LOC} is not valid in this model.

We have the following two theorems.

\begin{theorem}[Soundness and Completeness]\label{2Thm:Comp}
	The axiom schema $\textbf{K45+NU}$, plus the axioms of propositional logic and the inference rules necessitation and modus ponens are sound and complete with respect to simplicial belief models.
\end{theorem}
\begin{proof}
	Appendix \ref{2prf:comp}
\end{proof}

A reasonable question to ask is why the seriality schema, \textbf{D}, is not sound. This is because it is possible for an $a$-colored perspective, call it $n$, to be such that there is no $F\in\F(S_a)$ such that $n\in F$. We will call such perspectives \defin{isolated}. It is easy enough to assume that no perspectives are isolated, in which case $D$ is indeed sound. However, our revision concept will take as input models with no isolated perspectives, and create as output models with isolated perspectives. We will explore this in detail in Section \ref{2sec:revform}.

\section{The Revision Concept} \label{2sec:revinf}

Belief revision is certainly not a new idea in the realm of formal philosophy. \cite{Grove,Lewis2,AGM,MAGM} The basic idea is quite simple. When an agent learns $\varphi$, rather than simply eliminate all those possible worlds which she formally considered possible, but now are ruled out because they satisfied $\neg\varphi$, she instead may replace some of those eliminated worlds with ``nearby'' worlds that satisfy $\varphi$. For example, suppose agent $a$ considers two worlds possible, $w_1$ and $w_2$. Suppose there are two worlds she does not consider possible, $w_3$ and $w_4$. Let $P$ be true at $w_1$, $w_3$, and $w_4$, and $Q$ be true at $w_2$ and $w_3$. Now suppose $a$ learns that $P$ is true. She now no longer considers $w_2$ possible. However, under revision, instead of merely eliminating $w_2$, she may replace it with a ``nearby'' world that satisfies $P$. Both $w_3$ and $w_4$ are candidates for replacement, then. However, one might further argue that $w_3$ is ``nearer'' to $w_2$ than $w_4$, as both $w_2$ and $w_3$ make true $Q$. In this case, after learning $P$, agent $a$ would switch from considering the set $\{w_1,w_2\}$ possible to considering $\{w_1,w_3\}$ possible. Call this Scenario 1. If, instead, one considered $w_3$ and $w_4$ equally near to $w_2$, then after learning $P$ she would consider $\{w_1,w_3,w_4\}$ possible. Call this Scenario 2. 

Many variations have been given on this basic premise. In the initial presentation by David Lewis, one can specify a ``nearness'' function that takes any world/formula pair to the nearest set of worlds satisfying the given formula.\cite{Lewis2} More specifically, $f:W\times\mathbf{Form}\rightarrow 2^W$ is a nearness function, where $\mathbf{Form}$ is the set of formulas.\footnote{This is not the notation that Lewis uses but it will be clearer for our purposes.} Under all forms of revision, if an agent starts out believing in the set of worlds $X$, then after learning $\varphi$, they believe the set $\bigcup_{w\in X}f(w,\varphi)$. In Scenario 1 above, $f(w_2,P)=\{w_3\}$ and $f(w_1,P)=\{w_1\}$.\footnote{It is typicially assumed that if a world $w$ satisfies a formula $\varphi$, $f(w,\varphi)=w$.} So, the agent revises her beliefs from $X=\{w_1,w_2\}$ to $\bigcup_{w\in X}f(w,P)=\{w_1,w_3\}$ In Scenario 2, $f(w_2,P)=\{w_3,w_4\}$ and $f(w_1,P)=\{w_1\}$. So, the agent revises her beliefs from $X=\{w_1,w_2\}$ to $\bigcup_{w\in X}f(w,P)=\{w_1,w_3,w_4\}$. Importantly, in this initial presentation, Lewis makes no restrictions on the nearness function. 

Grove, by contrast, does make some restrictions on nearness functions. \cite{Grove} For Grove, nearness is given by a system of nested spheres on the set of worlds $W$\footnote{Grove assumes that the set of worlds is the set of all maximal consistent sets of formulas. We will present his theory in a more general context here.}. More specifically, $\mathbf{S}\subseteq 2^W$ is a system of nested spheres if and only if:\footnote{$|\varphi|$ is the set of worlds satisfying $\varphi$.}

\begin{enumerate}
	\item $\mathbf{S}$ is totally ordered under subset.
	\item There is a nonempty set $X\in\mathbf{S}$ such that $X$ is minimal under the subset ordering.
	\item $W\in\mathbf{S}$
	\item For any formula $\varphi$, there is a set $V\in\mathbf{S}$ such that $|\varphi|\cap V\neq\emptyset$ and for all $U\in\mathbf{S}$ such that $|\varphi|\cap U\neq\emptyset$, $V\subseteq U$. That is, for any formula $\varphi$, there is a subset minimal sphere $V$ intersecting that formula's satisfaction set.\footnote{As before, this is not given in the original notation, but is clearer for our purposes.}
\end{enumerate}

The reader should note that the minimal set under the ordering, $X$, is conceptually important. $X$ is the set of worlds the agent starts out believing, before she learns any new information. From this, Grove defines the notion of the ``closest'' worlds in $W$ to $X$ in which a formula $\varphi$ holds. Let $\text{min}_\varphi(\mathbf{S})$ be the smallest ``sphere'' in $\mathbf{S}$ which contains worlds satisfying $\varphi$. That is, $\text{min}_\varphi(\mathbf{S})$ is the unique $V\in\mathbf{S}$ such that $|\varphi|\cap V\neq\emptyset$ and for all $U\in\mathbf{S}$, then if $|\varphi|\cap U\neq\emptyset$, then $V\subseteq U$. With this in hand, we say that the closest worlds in $W$ to $X$ satisfying $\varphi$ are given by 

$$f'(X,\varphi)=|\varphi|\cap\text{min}_\varphi(\mathbf{S})$$

That is, $f'(X,\varphi)$ is the set $V$ as specified in the 4th property of a system of nested spheres. $f'$ is not a nearness function as we have defined it, however. Ultimately, Grove appeals to the following nearness function:

$$f(w,\varphi)=\begin{cases}
	X\cap|\varphi|	& \text{if}\;X=f'(X,\varphi) \\
	f'(X,\varphi)	& \text{if}\;X\neq f'(X,\varphi)
\end{cases}$$

Note that $X=f'(X,\varphi)$ if and only if $X\cap|\varphi|\neq\emptyset$. So, the specified nearness function says that if the agent believes $\varphi$ is possible, the nearest world to $w$ is always those worlds in $X$ satisfying $\varphi$, namely $X\cap\varphi$. Otherwise, the agent shifts to those worlds satisfying $\varphi$ in the smallest sphere that contains such worlds, i.e., $f'(X,\varphi)$. This ensures that the agent does not \emph{revise} unless she rules out \textit{all} of the worlds she initially considers possible. By ``revise'', we mean specifically the process of considering new worlds possible which she did not consider possible before. This is distinct from merely ``updating'', which could be revision, or simply removing those worlds contradicting what was just learned. So, in the Grove model, the agent ``updates'' each time she learns $\varphi$, but she only ``revises'' when $X\cap|\varphi|=\emptyset$. By contrast, in the more general Lewis setting, the agent may revise any time she learns a new proposition.

Let's look at an example of the Grove model in action. Recall our story from before. Agent $a$ considers two worlds possible, $w_1$ and $w_2$. There are two worlds she does not consider possible, $w_3$ and $w_4$. $P$ is true at $w_1$, $w_3$, and $w_4$, and $Q$ is true at $w_2$ and $w_3$. Since $X={w_1,w_2}$, the initial set of possible worlds, is such that $X\cap|P|\neq\emptyset$, $f(w_1,P)=f(w_2,P)=X\cap|P|=\{w_1\}$, giving that after learning $P$ the agent believes in the singleton set of worlds $\{w_1\}$. This captures the intuition that the agent should not revise her beliefs unless what she learns \textit{contradicts} her beliefs.\footnote{This is the third positive axiom (+3) of the ``AGM'' postulates, as presented in \cite{AGM}. Of course, Grove shows that nested sphere models satisfy all of the 8 positive AGM postulates. \cite{Grove}} This presents us with two different intuitions for revision definitions. On the one hand, an agent can revise every time she learns a new piece of information, seeking to replace any world which she has ruled out. On the other hand, an agent only revises if she has ruled out every world she initially considers possible. Ultimately we will provide a notion of revision which can be modified to be compatible with either intuition.

Both of these examples nicely illustrate what we take to be the biggest issue with belief revision as a model of learning. That is, in practice, it is very difficult to know which worlds should be considered nearer to any other. Lewis, of course, makes no restrictions on nearness whatsoever. While Grove does make restrictions on nearness, they do not tell a modeler which worlds should go into which spheres in a practical sense. If one were to attempt to apply Grove's model of revision in practice, no clues seem to be given as to which worlds are nearer to any other. Some of this is by design. For example, when using revision to give a semantics for counterfactuals or other natural language constructs, one can select worlds to be near to each other precisely to get the intended speaker intuition about a sentence to fall out. \cite{Stalnaker1981} However, when one does not have such speaker intuition, or any other clear notion about which worlds should be nearer to each other, it remains difficult to practically use revision.

The ultimate goal of this paper will be to motivate a particular concept of ``nearness'' which is very easily modeled in simplicial semantics. We should not overstate our aims here, however. While we will motivate what we hope to be an intuitive notion of nearness, we do not believe this is the only, or even best, notion of nearness which one could appeal to. Nor do we believe that this notion of nearness will apply in every situation one can imagine. It will suffice for us that this notion of nearness applies in a wide range of contexts. All we are hoping to demonstrate to the reader is a particular intuition which seems reasonable enough that it is worth following up on in a formal setting.

Consider the following case: You are agent $a$. You have three friends, call them agents $b$, $c$, and $d$. Therefore the facets in this case are tetrahedrons with four perspectives total, one for each agent. There are three facts, call them $B$, $C$, and $D$. The fact $X$, should be read as ``$x$ is at the park'', hence $B$ means that ``$b$ is at the park'' and so on.

Initially, you consider a single world/facet possible, consisting of your perspective, and a perspective for each of your friends. We will label your perspective $a'$. The perspectives of each of your friends, which we will label $x_1$, consist of the fact that $X$ is true, and no other facts. So, for example, the $b$-perspective $b_1$ is such that $L(b_1)(B)=1$, and $L(b_1)(C)=L(b_1)(D)=2$. Ambiently, there are additional perspectives for your friends, call them $x_0$ for each friend $x$. These are such that $X$ is false, and no other facts are true or false. Or, more specifically, for the $b$-perspective $b_0$, $L(b_0)(B)=0$, and $L(b_0)(C)=L(b_1)(D)=2$. As we said before, initially, you only consider one world/facet possible, and this is the tetrahedron $\{a',b_1,c_1,d_1\}$. Hence, you believe that each of your friends is at the park.

Now you learn the following fact: Either $b$ and $c$ are both not at the park, $d$ is not at the park, or all three are not at the park. This is the proposition $(\neg B\wedge\neg C)\vee\neg D$. Since this contradicts your beliefs, you no longer consider any worlds possible, and hence you must revise your beliefs. How should you do so? Let's consider some possible candidates for revision. Intuitively, these worlds should preserve your perspective, since we interpret facts at your perspective as indefeasible as you learn things. These are the candidate replacement facets for the contradicted facet $X$:

\begin{itemize}
	\item $X_1:=\{a',b_1,c_1,d_0\}$
	\item $X_2:=\{a',b_0,c_0,d_1\}$
	\item $X_3:=\{a',b_0,c_0,d_0\}$
\end{itemize}

Which of these is the best replacement? We want to argue that $X_1$ is a better candidate than $X_2$, and $X_2$ is a better candidate than $X_3$, corresponding to the fact that $|X\cap X_1|>|X\cap X_2|>|X\cap X_3|$. $X_1$ and $X$ share ``more in common'' with each other than $X$ and $X_2$. Indeed, if worlds are made of perspectives, as is the conceit of simplicial semantics, then $X$ and $X_1$ sharing more perspectives than $X$ and $X_2$ becomes a simple metric by which we can say $X$ and $X_1$ are ``more similar'' or ``nearer''. In practice, this consists of wishing to preserve as many of your fellow agents' perspectives, or points of view, as possible. So, in this example, $X_1$ is nearer than $X_2$ because $X_1$ preserves both $b$'s and $c$'s perspectives, while $X_2$ preserves only $d$'s perspective. $X_3$, of course, preserves nothing.

This example is fairly extreme, however. This is a case of revision where you are revising your beliefs in light of all worlds you consider possible being ruled out. As mentioned above, we want to allow for less forced cases of revision, where you may revise even if you still consider some worlds possible. In the new setup, you initially consider two worlds possible, $X$ and $X_3$. Note that this means you are initially uncertain about all three of $B$, $C$, and $D$. You then learn that $d$ is at the park, that is, you learn that $D$. This is of course consistent with $X$, but rules out $X_3$. So, what are the candidates for replacing $X_3$? As before, they must contain $a'$, and they should make true $D$. This gives us the following candidates:

\begin{itemize}
	\item $X_2:=\{a',b_0,c_0,d_1\}$
	\item $X_4:=\{a',b_1,c_0,d_1\}$
	\item $X_5:=\{a',b_0,c_1,d_1\}$
	\item $X:=\{a',b_1,c_1,d_1\}$
\end{itemize}

Note that $|X_3\cap X_2|>|X_3\cap X_4|=|X_3\cap X_5|>|X_3\cap X|$. If we use our intuition from previously, that preserving more perspectives gives the best notion of nearness, then $X_2$ is our best candidate replacement world. This also has a few other niceties. We learn $D$, which in this setup means we learn information about $d$'s perspectives. In this setup, there is no overlap between the perspectives of your friends. Each perspective corresponds uniquely to a specified literal. So, choosing $X_2$ as the replacement facet for $X_3$ means following the intuition to minimally alter your opinion of $b$'s perspective and $c$'s perspective when learning about $d$'s perspective. In this case, you preserve your uncertainty about the facts $B$ and $C$ when you learn $D$. There is another intuition, however, which suggests that in this case, you should actually pick $X$ as the replacement for $X_3$, despite the fact that it shares the \textit{fewest} perspectives with $X_3$ out of all the possible candidates. This is because you already consider $X$ possible. In terms of uncertainty, this choice captures the fact that your initial uncertainty about $B$, $C$, and $D$ are tied together. Really, you believe that all three facts are true, or all three are false. So, if you learn that one fact is true, based on your initial setup, you infer that all three are true. This, of course, corresponds with the usual Bayesian update on new information. Under this intuition, revision would be reserved \textit{only} for the case where \textit{every} world you consider possible is ruled out. We will provide formal definitions that capture both of the above intuitions.

However, all of these examples are restricted to the case where there is a one-to-one correspondence between perspectives and literals.\footnote{As such, all of the above examples may be given in the language of local variables. \cite{KaSC}} Hence, we could equally motivate all of the above examples by arguing we are saying worlds are nearer the more literals they validate in common, rather than the more perspectives they share. We should check that our ``preserving perspectives'' notion makes sense in an example where this is not true. As before, your perspective is $a'$. Following our example from the introduction, we will imagine a case where our friends are at the park, and they are observing the status of a pond at said park. Let $D$ stand for ``there are ducks in the pond'' and $F$ stand for ``there are frogs in the pond''. We set up the perspectives as follows. For any $x\in\{b,c,d\}$, the perspective $x_{i,j}$ is such that $L(x_{i,j})(D)=i$ and $L(x_{i,j})(F)=j$. So, $b_{1,0}$ is the $b$-perspective where $b$ is seeing that there are ducks in the pond, and there are no frogs in the pond. By contrast, $c_{1,2}$ is the $c$-perspective where $c$ is seeing that there are ducks in the pond, but they are not seeing either that there are or are not frogs in the pond.

Initially, you consider two worlds possible. These are $X:=\{a',b_{1,1},c_{1,1},d_{1,1}\}$, and $Y:=\{a',b_{1,0},c_{1,2},d_{1,2}\}$. In the first world, all of your friends see ducks and frogs. In the second world, all of your friends see ducks, while $b$ sees that there are no frogs, and $c$ and $d$ neither see frogs nor see that there are no frogs. Then you learn that $F$, which contradicts $Y$. One can easily check that the nearest candidate replacement world is given by $Y':=\{a',b_{1,1},c_{1,2},d_{1,2}\}$. At this world, $D$ and $F$ are both true. But what about $Y'':=\{a',b_{1,2},c_{1,1},d_{1,2}\}$? This world satisfies the same literals as $Y'$, making true both $D$ and $F$. However, $1=|Y\cap Y'|>|Y\cap Y''|=2$. This is because $Y$ differs from $Y'$ only in $b$'s perspective, but $Y''$ differs from $Y$ in both $b$ and $c$'s perspective. This motivates the notion of nearness best. $Y$ was a world making true $D$ and $\neg F$. However, $b$'s perspective was the reason that $F$ was false at $Y$. So, when we learn that $F$ is actually true, when finding nearby candidate worlds, we should focus on changing $b$'s perspective when building our replacement for $Y$. 

This ``building'' intuition is useful to elaborate on. When imagining candidate replacement worlds, the following story is useful. Revision has, at its basis, the idea that we replace worlds with nearby ones on a world-by-world basis. That is, it doesn't (necessarily) matter that $w$ is related to $w'$ when determining nearness. So, when $Y$ is ruled out by learning $F$, we can focus on finding a replacement for $Y$ in isolation. When we look at \textit{why} $Y$ is ruled out by $F$, we find, in this case, that the culprit is $b$'s perspective. So, when searching for a replacement, we start by swapping out other $b$-perspectives in until we get something that works. This is the sense in which we are ``building'' the replacement(s) for $Y$. As it turns out, in this example, changing $b$'s perspective alone is sufficient. If it weren't, we could, intuitively, \textit{then} try swapping out $c$ perspectives and $d$ perspectives until we found something that worked. But, since swapping the $b$ perspective is a forced choice, we try that alone first. Only if that fails do we try swapping more perspectives. Hence, sharing more perspectives gives us our notion of nearness.

A final note is that the notion of a formula being true at a particular agent's perspective is not formalized directly in our language. For example, $Y''$ and $Y'$ satisfy exactly the same non-modal formulas. However, doing so would seems to be simple. One could add a modal operator $A_x$ for all $x\in Ag$ which applies only to literals, and for a model $\M$ and facet $X$, we could say that $\M,X\vDash A_xP$ iff $L(\pi_x(X))(P)=1$, and $\M,X\vDash A_x\neg P$ iff $L(\pi_x(X))(P)=0$. It would follow that $\M,Y'\vDash A_bP$, but $\M,Y''\nvDash A_bP$. 

It seems that this could be a fruitful connection between simplicial semantics and something like awareness logic.\cite{Awareness} One can read $A_x(P)$ as stating that ``$x$ is aware of $P$'', and as a consequence think of the literals true at a perspective as constituting what the agent is ``aware'' of at that perspective. Because what is and is not assigned to perspectives, and hence what agents are and are not aware of, plays such a major role in our conception of revision, this would seem to be a fruitful direction for future work. However, it is beyond the scope of this paper.

\section{Simplicial Semantics for Belief Revision} \label{2sec:revform}

In order to formalize the intuition set out in Section \ref{2sec:revinf}, we will turn to an example. Let $Ag=\{a,b,c\}$, and suppose each agent has a privately held bit value, denoted with the proposition $P_i$ for $i\in Ag$. $P_i$ denotes that $i$'s bit value is 1, and $\neg P_i$ denotes that $i$'s bit value is 0. Because the value is privately held, $i$ is aware of the truth value of $P_i$ (it is in $i$'s perspective) but no agent other than $i$ is aware of the truth value of $P_i$. This gives us six perspectives, two for each agent. We model this formally as $N=\{a_0,a_1,b_0,b_1,c_0,c_1\}$, $V(i_j)=i$ and $L(i_j)(P_i)=j$ for all $i\in Ag$ and $j\in\{0,1\}$. These are depicted in Figure \ref{2fgr:mainrevpersp}.

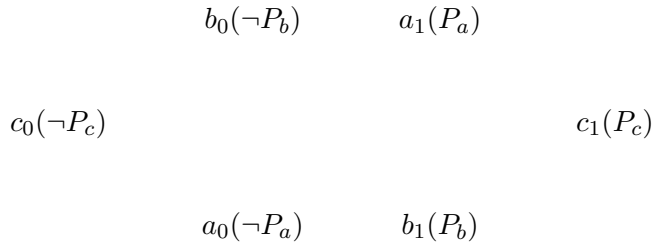
\begin{figure}[htbp]
	\begin{center}
		\begin{tikzcd}
			& b_0(\neg P_b)  & a_1(P_a)   &                                     \\
			c_0(\neg P_c) &                                                        &                                                         & c_1(P_c) \\
			& a_0(\neg P_a)                                          & b_1(P_b)                                              
		\end{tikzcd}
		\caption{6 perspectives modelling 3 agents each with a privately held bit value.} \label{2fgr:mainrevpersp}
	\end{center}
\end{figure}

Now, assume that for all $i\in Ag$, $\mathcal{F}(S_i)=\{\{a_1,b_1,c_1\},\{a_1,b_1,c_0\},\{a_0,b_0,c_0\}\}$. We will call these facets $w_0=\{a_0,b_0,c_0\}$, $w_1=\{a_1,b_1,c_1\}$, and $w_2=\{a_1,b_1,c_0\}$ for brevity. The choice of worlds in this example is arbitrary except insofar as they will be illustrative when considering our revision mechanism. This is drawn in Figure \ref{2fgr:mainrevstart}.

\begin{figure}[htbp]
	\begin{center}
		\begin{tikzcd}
			& b_0(\neg P_b) \arrow[ld, no head] \arrow[dd, no head] & a_1(P_a) \arrow[rd, no head] \arrow[dd, no head]  &                                     \\
			c_0(\neg P_c) \arrow[rd, no head] \arrow[rru, no head] &                                                        &                                                         & c_1(P_c) \arrow[ld, no head] \\
			& a_0(\neg P_a)                                          & b_1(P_b) \arrow[llu, no head]                                             
		\end{tikzcd}
		\caption{The model of 3 agents with privately held bit values prior to announcement. The three triangles correspond to the three worlds.} \label{2fgr:mainrevstart}
	\end{center}
\end{figure}
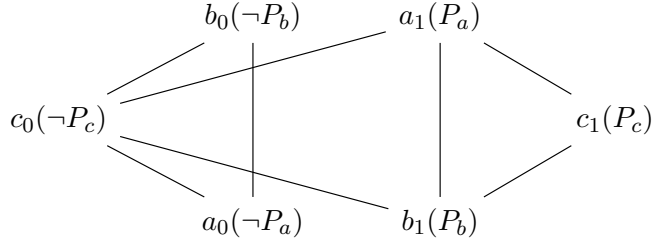

Now we suppose that $P_a$ is publicly announced. How should we imagine each agent learning this piece of information? It is useful to think perspective by perspective. First, imagine you are agent $a$. Figure \ref{2fgr:mainrevstart} represents your epistemic frame prior to learning any information. If you are in perspective $a_1$, you consider two worlds possible: $w_1$ and $w_2$. Both of these worlds make true $P_a$, so you do not cease to consider either possible after announcement. However, suppose you are in perspective $a_0$. What has just been publicly announced, $P_a$, contradicts what is true at the only world you consider possible, namely $w_0$. So, you no longer consider this world possible. Moreover, revision is not possible. Candidate revision worlds should at minimum preserve your perspective. Since at $a_0$, $\neg P_a$ is in your perspective, any world containing $a_0$ will contradict $P_a$. This update is drawn in Figure \ref{2fgr:mainreva}. Note that agent $a$'s update is the usual public announcement update. Worlds which contradict the announcement are removed with no revision.

\begin{figure}[htbp]
	\begin{center}
		\adjustbox{scale=.85}{
			\begin{tikzcd}
				& b_0(\neg P_b) \arrow[red, ld, no head] \arrow[red, dd, no head] & a_1(P_a) \arrow[red, rd, no head] \arrow[red, dd, no head]  &                                     \\
				c_0(\neg P_c) \arrow[red, rd, no head] \arrow[red, rru, no head] &                                                        &                                                         & c_1(P_c) \arrow[red, ld, no head] \\
				& a_0(\neg P_a)                                          & b_1(P_b) \arrow[red, llu, no head]                                             
			\end{tikzcd}\huge$\overset{P_a}{\Rightarrow}$\normalsize \begin{tikzcd}
				& b_0(\neg P_b) & a_1(P_a) \arrow[red, rd, no head] \arrow[red, dd, no head]  &                                     \\
				c_0(\neg P_c) \arrow[red, rru, no head] &                                                        &                                                         & c_1(P_c) \arrow[red, ld, no head] \\
				& a_0(\neg P_a)                                          & b_1(P_b) \arrow[red, llu, no head]                                             
		\end{tikzcd}}
		\caption{$S_a$ before and after $P_a$ is publicly announced}\label{2fgr:mainreva}
		
	\end{center}
\end{figure}
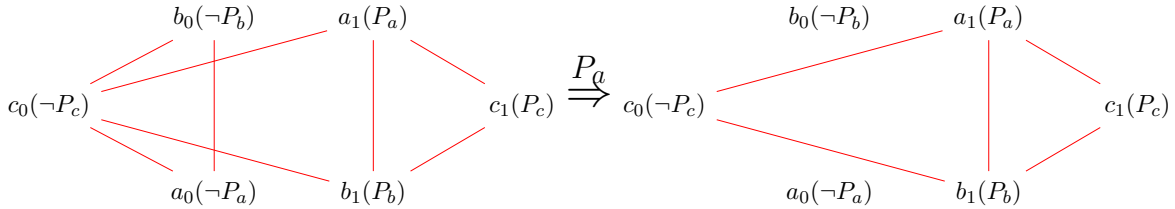

Things are trickier for agent $b$. Again, it is useful to consider things perspective by perspective. Suppose you are in perspective $b_1$. Then you consider two worlds possible, those being $w_1$ and $w_2$. Since both of these worlds satisfy $P_a$, neither is removed or changed after the announcement. Now suppose you are in perspective $b_0$. Before the announcement, you consider only world $w_0$ possible. This is then contradicted by the announcement, and so must be removed. However, unlike agent $a$, $b$ has the possibility to revise. As mentioned before, candidates for revision should preserve $b$'s point of view. Moreover, they need to satisfy what was just announced, $P_a$. The third condition is the one motivated in Section \ref{2sec:revinf}. Namely, replacement worlds should preserve as many perspectives from the replaced world as possible. In this case, the world to be replaced is $w_0=\{a_0,b_0,c_0\}$. By the first condition, candidate replacement worlds need to contain $b_0$. By the second condition, they need to satisfy $P_a$. In our example, that means candidate replacement worlds for $w_0$ after $P_a$ is announced must contain $a_1$. This gives us two potential candidates: $w_3=\{a_1,b_0,c_0\}$ and $w_4=\{a_1,b_0,c_1\}$. Note that $|w_3\cap w_0|>|w_4\cap w_0|$. Hence, $w_3$ and $w_0$ share more perspectives, namely two, than $w_4$ and $w_0$, which only share one perspective. So, from our informal intuition, $w_3$ is the best candidate replacement for $w_0$ after $P_a$ is announced. This is shown in Figure \ref{2fgr:mainrevb}.

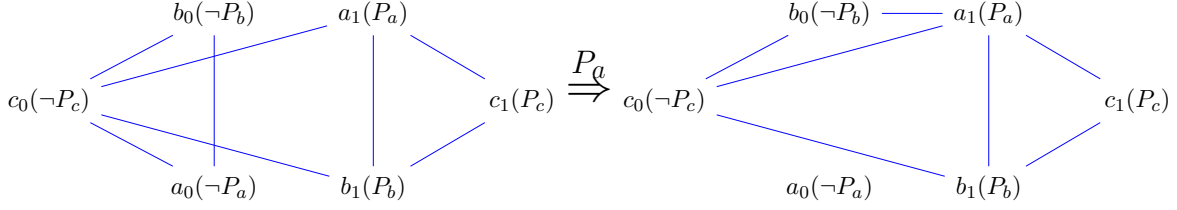
\begin{figure}[htbp]
	\begin{center}
		\adjustbox{scale=.85}{
			\begin{tikzcd}
				& b_0(\neg P_b) \arrow[blue, ld, no head] \arrow[blue, dd, no head] & a_1(P_a) \arrow[blue, rd, no head] \arrow[blue, dd, no head]  &                                     \\
				c_0(\neg P_c) \arrow[blue, rd, no head] \arrow[blue, rru, no head] &                                                        &                                                         & c_1(P_c) \arrow[blue, ld, no head] \\
				& a_0(\neg P_a)                                          & b_1(P_b) \arrow[blue, llu, no head]                                             
			\end{tikzcd}\huge$\overset{P_a}{\Rightarrow}$\normalsize \begin{tikzcd}
				& b_0(\neg P_b) \arrow[blue, r, no head] & a_1(P_a) \arrow[blue, rd, no head] \arrow[blue, dd, no head]  &                                     \\
				c_0(\neg P_c) \arrow[blue, rru, no head] \arrow[blue, ru, no head] &                                                        &                                                         & c_1(P_c) \arrow[blue, ld, no head] \\
				& a_0(\neg P_a)                                          & b_1(P_b) \arrow[blue, llu, no head]                                             
		\end{tikzcd}}
		\caption{$S_b$ before and after $P_a$ is publicly announced}\label{2fgr:mainrevb}
		
	\end{center}
\end{figure}

The last agent to explore is agent $c$. Suppose you are in perspective $c_1$. Before the announcement, you consider one world possible, namely $w_1$. This is not contradicted by the announcement, so there is no change. Now suppose you are in perspective $c_0$. You consider two worlds possible, namely $w_0$ and $w_2$. Since $w_0$ satisfied $\neg P_a$, this is ruled out by the announcement of $P_a$. Moreover, the same revision calculation carried out for agent $b$ reveals that the best candidate for $w_0$ is $w_3$. This is drawn in Figure \ref{2fgr:mainrevc}.

\begin{figure}[htbp]
	\begin{center}
		\adjustbox{scale=.85}{
			\begin{tikzcd}
				& b_0(\neg P_b) \arrow[green, ld, no head] \arrow[green, dd, no head] & a_1(P_a) \arrow[green, rd, no head] \arrow[green, dd, no head]  &                                     \\
				c_0(\neg P_c) \arrow[green, rd, no head] \arrow[green, rru, no head] &                                                        &                                                         & c_1(P_c) \arrow[green, ld, no head] \\
				& a_0(\neg P_a)                                          & b_1(P_b) \arrow[green, llu, no head]                                             
			\end{tikzcd}\huge$\overset{P_a}{\Rightarrow}$\normalsize \begin{tikzcd}
				& b_0(\neg P_b) \arrow[green, r, no head] & a_1(P_a) \arrow[green, rd, no head] \arrow[green, dd, no head]  &                                     \\
				c_0(\neg P_c) \arrow[green, rru, no head] \arrow[green, ru, no head] &                                                        &                                                         & c_1(P_c) \arrow[green, ld, no head] \\
				& a_0(\neg P_a)                                          & b_1(P_b) \arrow[green, llu, no head]                                             
		\end{tikzcd}}
		\caption{$S_c$ before and after $P_a$ is publicly announced, assuming revision happens to every eliminated world.}\label{2fgr:mainrevc}
		
	\end{center}
\end{figure}
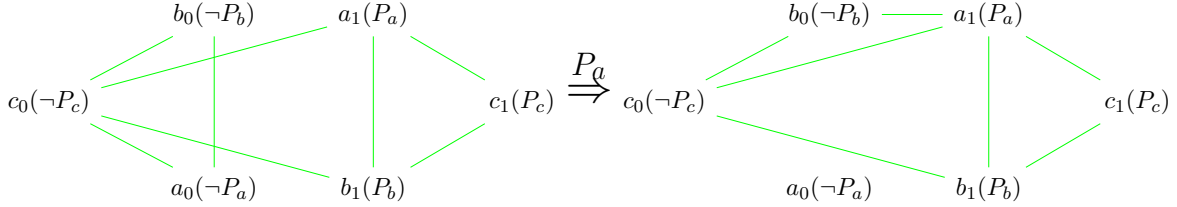

However, this is not the only intuition one could have about how $c$ should revise. When $b$ revised around perspective $b_0$, it is notable that all of the worlds containing $b_0$ had been eliminated by the announcement. By contrast, both there are worlds consistent with the announcement containing both $c_1$, and in particular, $c_0$, namely $w_2$. Because $c$ already considers $w_2$ possible, it is reasonable to argue that $w_2$ is ``nearer'' to $w_0$, even though $|w_3\cap w_0|>|w_2\cap w_0|$. This would track the intuition that revision should only happen if all worlds the agent considers possible are eliminated, as per Grove's nested sphere model, discussed in Section \ref{2sec:revinf}. \cite{Grove} In this case, the revision looks as it appears in Figure \ref{2fgr:mainrevcG}. 

\begin{figure}[htbp]
	\begin{center}
		\adjustbox{scale=.85}{
			\begin{tikzcd}
				& b_0(\neg P_b) \arrow[green, ld, no head] \arrow[green, dd, no head] & a_1(P_a) \arrow[green, rd, no head] \arrow[green, dd, no head]  &                                     \\
				c_0(\neg P_c) \arrow[green, rd, no head] \arrow[green, rru, no head] &                                                        &                                                         & c_1(P_c) \arrow[green, ld, no head] \\
				& a_0(\neg P_a)                                          & b_1(P_b) \arrow[green, llu, no head]                                             
			\end{tikzcd}\huge$\overset{P_a}{\Rightarrow}$\normalsize \begin{tikzcd}
				& b_0(\neg P_b) & a_1(P_a) \arrow[green, rd, no head] \arrow[green, dd, no head]  &                                     \\
				c_0(\neg P_c) \arrow[green, rru, no head] &                                                        &                                                         & c_1(P_c) \arrow[green, ld, no head] \\
				& a_0(\neg P_a)                                          & b_1(P_b) \arrow[green, llu, no head]                                             
		\end{tikzcd}}
		\caption{$S_c$ before and after $P_a$ is publicly announced, assuming revision happens only when a perspective is no longer contained in any world.}\label{2fgr:mainrevcG}
		
	\end{center}
\end{figure}
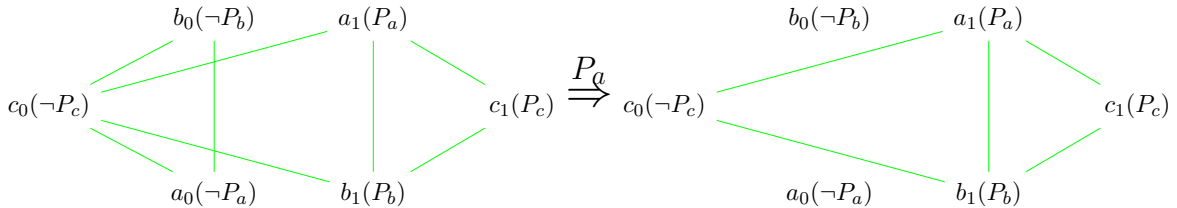

Fix a formula $\varphi$ which contains no modalities in its recursive construction. That is, $\varphi$ is a proposition. Following the intuitions above, given a model $\M$, if a facet $X\in\F(S_a)$ is ruled out by a public announcement of $\varphi$, candidate replacement facets $Y$ needs to satisfy the following three conditions:

\begin{enumerate}
	\item $X$ and $Y$ need to share the same $a$-perspective, or more formally, $\pi_a(X)=\pi_a(Y)$
	\item $Y$ need to satisfy $\varphi$, or more formally, $\M,Y\vDash\varphi$
	\item If $Z$ satisfies conditions $1$ and $2$, $Y$ needs to share at least as many perspectives with $X$ as $Z$, or more formally, $|Y\cap X|\geq|Y\cap Z|$.
\end{enumerate}

We can formalize these three conditions into the following revision rule:

\begin{definition}{$\mathcal{M}[\varphi]$}\label{2DefRev1}
	
	Consider a UCF simplicial model $\mathcal{M}=\langle N,V,L,\{S_a\}_{a\in A}\rangle$, and $\varphi$ a formula with no modalities in its recursive construction. Let  $N_\varphi:=N$, $V_\varphi:=V$, and $L_\varphi:=L$. Further, let $\mathbf{F}_\varphi$ be the subset of $\mathcal{F}(\mathfrak{M}((N_\varphi,V_\varphi,L_\varphi))$ which satisfy $\varphi$. We need to define a function $\mathcal{R}_a:\mathcal{F}(S_a)\times\L_{B}(Ag)\rightarrow2^{\mathbf{F}_\varphi}$.

	$$\mathcal{R}_a(F,\varphi):=\{G\in\mathbf{F}_\varphi~|~(\pi_a(F)=\pi_a(G))\wedge\forall X\in\mathbf{F}_\varphi((\pi_a(X)=\pi_a(G))\rightarrow(|F\cap G|\geq|X\cap G|))\}$$
	
	Note that $\mathcal{R}_a:\F(S_a)\times\L_{B}(Ag)\rightarrow 2^{\mathbf{F}_\varphi}$ is a nearness function.\footnote{The total set of worlds $S$ is $\mathfrak{M}((N_\varphi,V_\varphi,L_\varphi)=\mathfrak{M}(N,V,L)$, as $\F(S_a)\subseteq \mathfrak{M}(N,V,L)$ and $\mathbf{F}_\varphi\subseteq\mathfrak{M}(N_\varphi,V_\varphi,L_\varphi)$.} Further, note that, if $G\in\mathbf{F}_\varphi$, then $\mathcal{R}_a(G,\varphi)=\{G\}$. Let each $S_{a,\varphi}$ be the simplicial complex whose facets are elements of the set $\bigcup_{F\in\mathcal{F}(S_a)}\mathcal{R}_a(F,\varphi)$. Then $\mathcal{M}[\varphi]:=\langle N_\varphi,V_\varphi,L_\varphi,\{S_{a,\varphi}\}_{a\in A}\rangle$.
	
\end{definition}\footnote{This definition of revision, and the one following, also seems perfectly well motivated in settings where one uses the language of Local Variables. \cite{KaSC} So, if one prefers that setting, this definition should be conceptually compatible.}

If $S$ is not maximal, simply take each $S_{a,\varphi}$ be the simplicial complex whose facets are elements of the set $\F(S)\cap\bigcup_{F\in\mathcal{F}(S_a)}\mathcal{R}_a(F,\varphi)$. It is easy to check that the above definition gives a UCF sub-simplicial complex of the maximal complex of $N_\varphi$, $V_\varphi$, and $L_\varphi$ (or a UCF subcomplex of $S$ if $S$ is non-maximal). Therefore, the above definition gives us a model $\M[\varphi]$ which represents the result of a model $\M$ after each agent revises their beliefs in light of a public announcement of $\varphi$. 

A potential fourth condition is that, if there is already a facet $Y\in\F(S_a)$ such that $\pi_a(X)=\pi_a(Y)$ and $\M,Y\vDash\varphi$, then $Y$ is the nearest facet to $X$, regardless of intersection size. This potential fourth condition is motivated by the intuition that revision should only happen if a perspective becomes isolated, i.e., an $a$-perspective is no longer contained in a facet of $S_a$. We formalize this as follows:

\begin{definition}{$\mathcal{M}[\varphi]_G$\footnote{The subscript stands for ``Grove''.}}\label{2DefRev2}
	
	Consider a UCF simplicial model $\mathcal{M}=\langle N,V,L,\{S_a\}_{a\in A}\rangle$, and $\varphi$ a formula with no modalities in its recursive construction. Let $N_\varphi:=N$, $V_\varphi:=V$, and $L_\varphi:=L$. We need to define a function $\mathcal{R}_{a,G}:\mathcal{F}(S_a)\times\L_{B}(Ag)\rightarrow2^{\mathbf{F}_\varphi}$. Fix the set $\mathbf{T}_a(F,\varphi):=\{X\in\mathbf{F}_\varphi\cap\F(S_a)~|~\pi_a(X)=\pi_a(Y)\}$. Then
	
	$$\mathcal{R}_{a,G}(F,\varphi):=\begin{cases}
		\mathbf{T}_a(F,\varphi) & \text{if }\mathbf{T}_a(F,\varphi)\neq\emptyset\\
		\mathcal{R}_{a}(F,\varphi)	& \text{if }\mathbf{T}_a(F,\varphi)=\emptyset
	\end{cases}$$
	
	Note that $\mathcal{R}_{a,G}:\F(S_a)\times\L_{B}(Ag)\rightarrow 2^{\mathbf{F}_\varphi}$ is a nearness function. Further, note that, if $X\in\mathbf{F}_\varphi$, then $X\in\mathcal{R}_{a,G}(X,\varphi)$. Let each $S_{a,\varphi,G}$ be the simplicial complex whose facets are elements of the set $\bigcup_{F\in\mathcal{F}(S_a)}\mathcal{R}_a(F,\varphi)$. Then $\mathcal{M}[\varphi]_G:=\langle N_\varphi,V_\varphi,L_\varphi,\{S_{a,\varphi,G}\}_{a\in A}\rangle$
	
\end{definition}

As before, it is easy to check that the above definition gives a UCF sub-simplicial complex of the maximal complex of $N_\varphi$, $V_\varphi$, and $L_\varphi$. Therefore, the above definition gives us a model $\M[\varphi]_G$ which represents the result of a model $\M$ after each agent revises their beliefs in light of a public announcement of $\varphi$, where revision is restricted to the case where perspectives become isolated, represented by the case where $\mathbf{T}_a(F,\varphi)=\emptyset$. While we present both concepts here, for the remainder of the text, we will use $\M[\varphi]$ unless otherwise stated.

\subsection{Examples} \label{2subsec:ex}

The purpose of these examples is to become familiar with the updating procedure, and to motivate a particular issue that arises under certain interpretations of announcement. Figure \ref{2fgr:mainex1} captures the idea of ``everything going right.'' This is the kind of story we want to be able to capture.

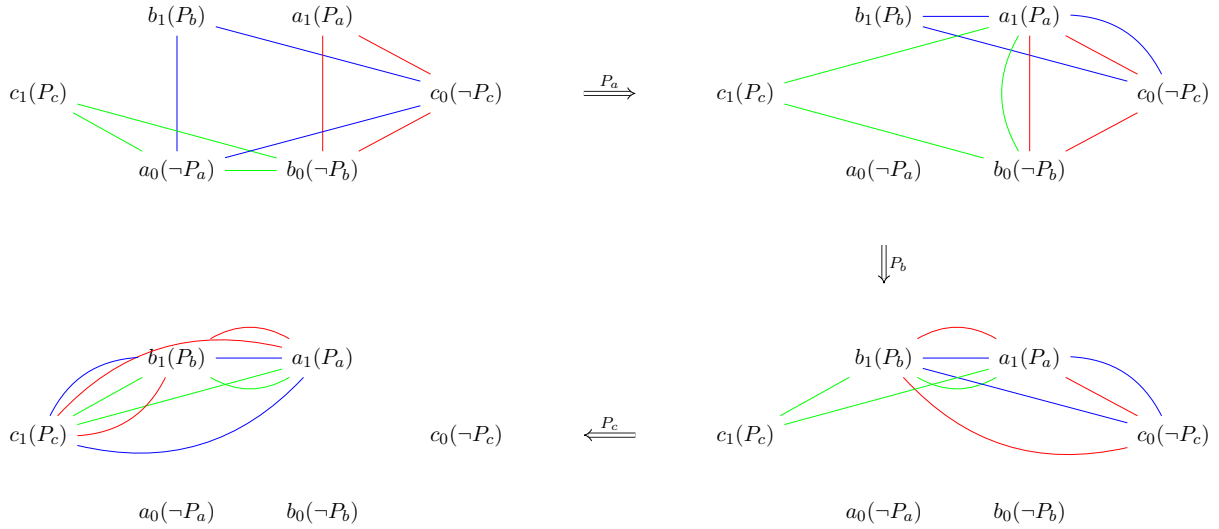
\begin{figure}[htbp]
	\begin{center}\adjustbox{scale=.75}{
			\begin{tikzcd}
				& b_1(P_b) \arrow[blue, dd, no head] \arrow[blue, rrd, no head]                                        & a_1(P_a) \arrow[red, dd, no head] \arrow[red, rd, no head]                                               &                  &                     &                      &                                                             & b_1(P_b) \arrow[blue, rrd, no head] \arrow[blue, r, no head]                                                                                                                   & a_1(P_a) \arrow[red, rd, no head] \arrow[blue, rd, no head, bend left] \arrow[red, dd, no head] \arrow[green, dd, no head, bend right] &                  \\
				c_1(P_c) \arrow[green, rd, no head] \arrow[green, rrd, no head]                                           &                                                                                                  &                                                                                                       & c_0(\neg P_c) & {} \arrow[r, "P_a", Rightarrow] & {}                   & c_1(P_c) \arrow[green, rrd, no head] \arrow[green, rru, no head] &                                                                                                                                                                            &                                                                                                    & c_0(\neg P_c) \\
				& a_0(\neg P_a) \arrow[blue, rru, no head] \arrow[green, r, no head]                                           & b_0(\neg P_b) \arrow[red, ru, no head]                                                                   &                  &                     &                      &                                                             & a_0(\neg P_a)                                                                                                                                                             & b_0(\neg P_b) \arrow[red, ru, no head]                                                                &                  \\
				&                                                                                                  &                                                                                                       &                  &                     &                      &                                                             & {} \arrow[d, "P_b", Rightarrow]                                                                                                                                                        &                                                                                                    &                  \\
				&                                                                                                  &                                                                                                       &                  &                     &                      &                                                             & {}                                                                                                                                                                         &                                                                                                    &                  \\
				& b_1(P_b) \arrow[blue, r, no head] \arrow[red, r, no head, bend left] \arrow[green, r, no head, bend right] & a_1(P_a) \arrow[green, lld, no head] \arrow[blue, lld, no head, bend left] \arrow[red, lld, no head, bend right] &                  &                     &                      &                                                             & b_1(P_b) \arrow[green, r, no head, bend right] \arrow[green, ld, no head] \arrow[blue, rrd, no head] \arrow[blue, r, no head] \arrow[red, r, no head, bend left] \arrow[red, rrd, no head, bend right] & a_1(P_a) \arrow[blue, rd, no head, bend left] \arrow[red, rd, no head]                                 &                  \\
				c_1(P_c) \arrow[green, ru, no head] \arrow[blue, ru, no head, bend left] \arrow[red, ru, no head, bend right] &                                                                                                  &                                                                                                       & c_0(\neg P_c) & {}                  & {} \arrow[l, "P_c"', Rightarrow] & c_1(P_c) \arrow[green, rru, no head]                      &                                                                                                                                                                            &                                                                                                    & c_0(\neg P_c) \\
				& a_0(\neg P_a)                                                                                   & b_0(\neg P_b)                                                                                       &                  &                     &                      &                                                             & a_0(\neg P_a)                                                                                                                                                             & b_0(\neg P_b)                                                                                    &                 
		\end{tikzcd}}
		\caption{Four models in a sequence of three announcements. The initial configuration depicts three agents, each of whose private bit value is 1, but has the mistaken belief that the other two agents have private bit value of zero. The announcements sequentially correct this misconception.}\label{2fgr:mainex1}
		
	\end{center}
\end{figure}

Suppose that the ``real'' world/facet is given by $\{a_1,b_1,c_1\}$. Then, at this facet, it is true that each agent has a true bit value of 1, but mistakenly believes that the other agents have a bit value of 0. Each public announcement of the true bit values brings the agents closer to consensus. There is no interesting higher order modal information here, as the different colored simplicial complexes do not interact with each other. Note that at each stage, there is exactly one red, one blue, and one green facet. Therefore, there can be no two worlds, or in our case facets, which are the same color and share a perspective of that color. Hence, there is no uncertainty about worlds/facets, and thus no higher order information.

Figure \ref{2fgr:mainex2} also captures ``things going well,'' but importantly motivates details about how we should think about ``candidate alternative'' for perspectives.

\begin{figure}[htbp]
	\begin{center}
		\adjustbox{scale=.95}{\begin{tikzcd}
				& b_1\color{black}(\neg P)         &                             &                              &    &                         & b_1\color{black}(\neg P)                            &                             \\
				{b_0\color{black}( P,Q)} & a \arrow[red, u, no head] & {b_2\color{black}(P,\neg Q)} & {} \arrow[r, "P", Rightarrow]            & {} & {b_0\color{black}( P,Q)} & a \arrow[red, r, no head] \arrow[red, l, no head] & {b_2\color{black}(P,\neg Q)} \\
				& {} \arrow[d, "P\wedge\neg Q"', Rightarrow]  &                             & {} \arrow[rrdd, "P\wedge Q", Rightarrow] &    &                         &                                                    &                             \\
				& {}                              &                             &                              &    &                         &                                                    &                             \\
				& b_1\color{black}(\neg P)         &                             &                              &    & {}                      & b_1\color{black}(\neg P)                            &                             \\
				{b_0\color{black}(P,Q)} & a \arrow[red, r, no head] & {b_2\color{black}(P,\neg Q)} &                              &    & {b_0\color{black}( P,Q)} & a \arrow[red, l, no head]                    & {b_2\color{black}(P,\neg Q)}
		\end{tikzcd}}
		\caption{Four models and three distinct updates based on three distinct announcements.}\label{2fgr:mainex2}
		
	\end{center}
\end{figure}
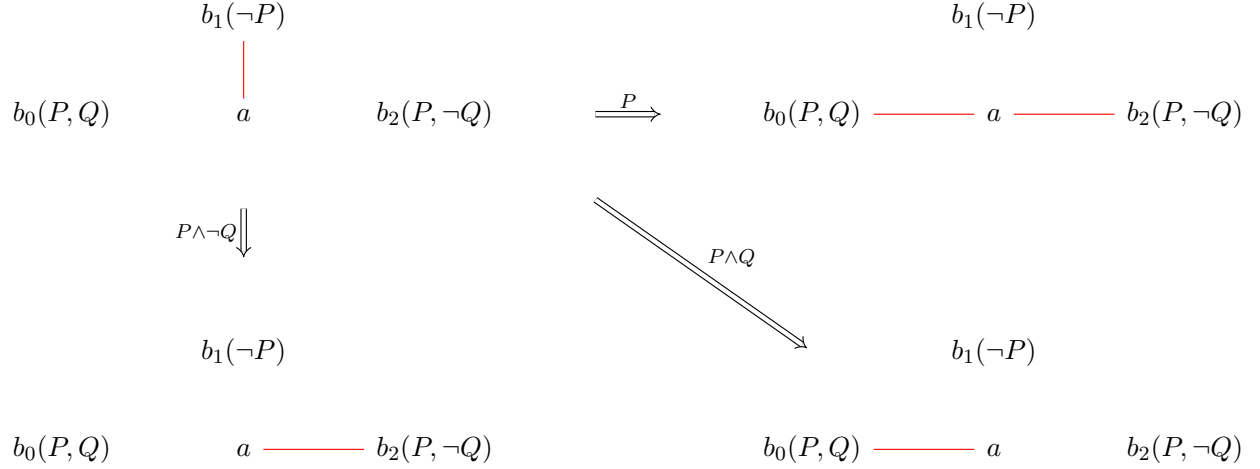

The complexes shown in Figure \ref{2fgr:mainex2} are all $S_a$. There are two key takeaways from this example. Firstly, what is announced significantly affects the update. When $P\wedge Q$ is announced, $a$ switches to the available perspective where both of these are true. When $P\wedge\neg Q$ is announced, a different perspective. Moreover, when only $P$ is announced, intuitively, $a$ should become uncertain, as there is no way for them to differentiate between $b_0$ and $b_2$ on the available information.

Figure \ref{2fgr:mainex3} motivates a flaw in this version of the formalism. Agents are, at present, very forgetful of what they've learned. If we take announcements to have a high degree of veracity, this seems like an issue that is worth overcoming.

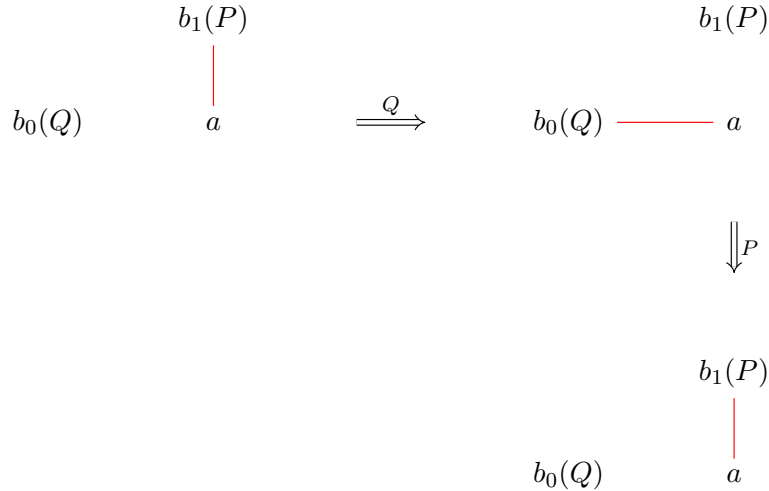
\begin{figure}[htbp]
	\begin{center}
		\begin{tikzcd}
			& b_1\color{black}(P) &                               &    &                                                    & b_1\color{black}(P)                    \\
			b_0\color{black}(Q) & a \arrow[red, u, no head] & {} \arrow[r, "Q", Rightarrow] & {} & b_0\color{black}(Q) \arrow[red, r, no head] & a                                       \\
			&                                 &                               &    &                                                    & {} \arrow[d, "P", Rightarrow]                      \\
			&                                 &                               &    &                                                    & {}                                                 \\
			&                                 &                               &    &                                                    & b_1\color{black}(P) \arrow[red, d, no head] \\
			&                                 &                               &    & b_0\color{black}(Q)                    & a                                      
		\end{tikzcd}
		\caption{Three models with two agents, showing an initial model with sequential updates after a public announcement of $Q$ then of $P$. Agent $a$ begins by believing $P$.}\label{2fgr:mainex3}
		
	\end{center}
\end{figure}

The complexes shown in Figure \ref{2fgr:mainex3} are all $a$'s complexes. In the example, $a$ in the first complex believes $P\wedge\neg Q$. Then $Q$ is announced, and $a$ believes $\neg P\wedge Q$. Then when $P$ is announced, the agent believes $P\wedge\neg Q$ again. In this sense, the agent is extremely forgetful of what has just been announced. The agent goes right back to believing $\neg Q$, even though they had just heard that $Q$ was announced.

One might think that the example in Figure \ref{2fgr:mainex3} is a bit forced. Yes, the agents are forgetful; the example is extremely constrained. Given the facets available to agent $a$, the two announcements are in fact inconsistent with each other. As we see in Figure \ref{2fgr:mainex4}, however, adding in more potential facets does not solve the issue.

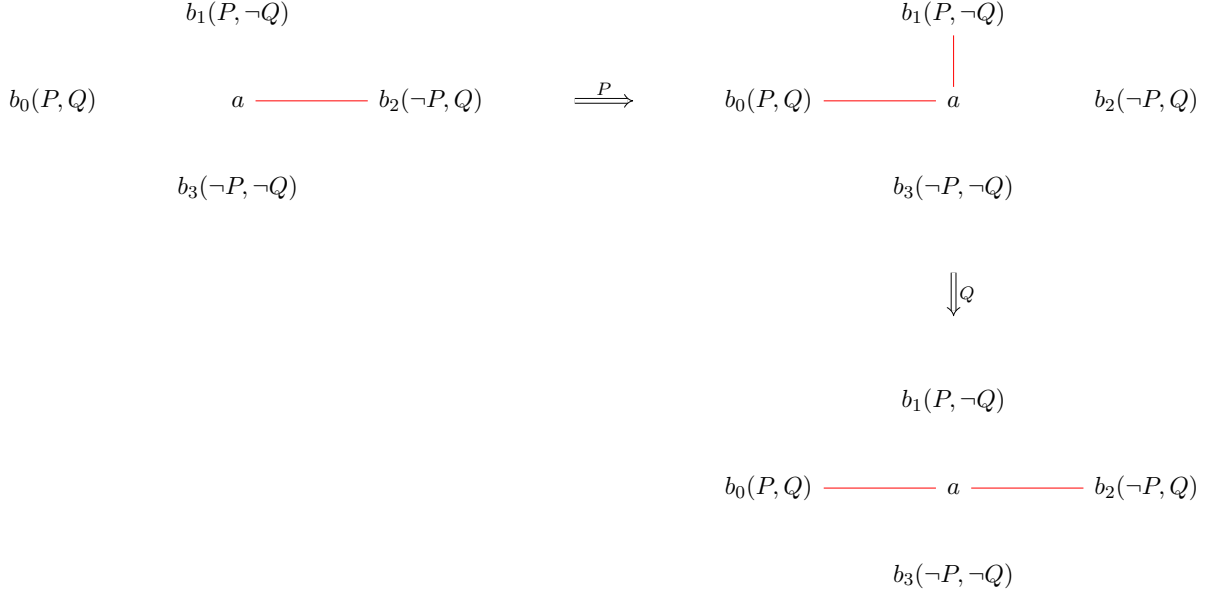
\begin{figure}[htbp]
	\begin{center}\adjustbox{scale=.85}{\begin{tikzcd}
				& b_1\color{black}(P,\neg Q)      &                                              &                               &    &                                         & b_1\color{black}(P,\neg Q)        &                                              \\
				b_0\color{black}(P,Q) & a \arrow[red, r, no head]                  & b_2\color{black}(\neg P,Q) & {} \arrow[r, "P", Rightarrow] & {} & b_0\color{black}(P,Q) & a \arrow[red, u, no head] \arrow[red, l, no head] & b_2\color{black}(\neg P,Q) \\
				& b_3\color{black}(\neg P,\neg Q) &                                              &                               &    &                                         & b_3\color{black}(\neg P,\neg Q)   &                                              \\
				&                                                  &                                              &                               &    &                                         & {} \arrow[d, "Q", Rightarrow]                      &                                              \\
				&                                                  &                                              &                               &    &                                         & {}                                                 &                                              \\
				&                                                  &                                              &                               &    &                                         & b_1\color{black}(P,\neg Q)        &                                              \\
				&                                                  &                                              &                               &    & b_0\color{black}(P,Q) & a \arrow[red, l, no head] \arrow[red, r, no head] & b_2\color{black}(\neg P,Q) \\
				&                                                  &                                              &                               &    &                                         & b_3\color{black}(\neg P,\neg Q)   &                                             
		\end{tikzcd}}
		\caption{Three models with two agents, showing an initial model with four facets in $W=\M(N,V,L)$, corresponding to all possible arrangements of the atoms $P$ and $Q$. As in Figure \ref{2fgr:mainex3}, there are two sequential updates. First the agent updates by $P$ and then $Q$. Agent $A$ begins by believing $\neg P\wedge Q$.}\label{2fgr:mainex4}
		
	\end{center}
\end{figure}

Figure \ref{2fgr:mainex4} is similar to Figure \ref{2fgr:mainex3}, except that there is a potential consistent facet for every combination of truth values of $P$ and $Q$. Note that, before any announcements, $a$ considers $\neg P\wedge Q$ possible. However, after the announcements of $P$ and $Q$, they should know $P\wedge Q$. Yet they consider $\neg P\wedge Q$ possible again! The situation is not ideal, to say the least.

\subsection{Memory} \label{2subsec:mem}

We will show a simple solution for this ``forgetfulness'' by using the $S$ complex as a kind of ``memory''. $S$ is the set from which all facets are selected already. In sequential updates, if subsequent $S$ keep a record that of everything that has been announced, then the candidate facets for revision, too, will keep such a record. We can define this formally (the new parts of the definition are highlighted in red):

\begin{definition}{$\mathcal{M}[\varphi]_m$}\label{2DefRevm}
	
	Consider a UCF simplicial model $\mathcal{M}=\langle N,V,L,\color{red}S,\color{black}\{S_a\}_{a\in A}\rangle$, and $\varphi$ a formula with no modalities in its recursive construction.. Let  $N_\varphi:=N$, $V_\varphi:=V$, $L_\varphi:=L$. Further, let $\mathbf{F}_\varphi$ be the subset of $\mathcal{F}(\mathfrak{M}((N_\varphi,V_\varphi,L_\varphi))$ which satisfy $\varphi$. \color{red} Also let $S_\varphi$ be the simplicial complex whose facets are given by $\F(S)\cap\mathbf{F}_\varphi$. \color{black} We need to define a function $\mathcal{R}_a:\mathcal{F}(S_a)\times\L_{B}(Ag)\rightarrow2^{\mathbf{F}_\varphi}$.

	$$\mathcal{R}_a(F,\varphi):=\{G\in\color{red}\F(S_\varphi)\color{black}~|~(\pi_a(F)=\pi_a(G))\wedge\forall X\in\mathbf{F}_\varphi((\pi_a(X)=\pi_a(G))\rightarrow(|F\cap G|\geq|X\cap G|))\}$$
	
	Note that $\mathcal{R}_a:\F(S_a)\times\L_{B}(Ag)\rightarrow 2^{\mathbf{F}_\varphi}$ is a nearness function. Further, note that, if $G\in\color{red}\F(S_\varphi)\color{black}$, then $\mathcal{R}_a(G,\varphi)=\{G\}$. Let each $S_{a,\varphi}$ be the simplicial complex whose facets are elements of the set $\bigcup_{F\in\mathcal{F}(S_a)}\mathcal{R}_a(F,\varphi)$. Then $\mathcal{M}[\varphi]_m:=\langle N_\varphi,V_\varphi,L_\varphi,\color{red} S_\varphi,\color{black}\{S_{a,\varphi}\}_{a\in A}\rangle$.
	
\end{definition}

Let us observe how this new definition solves our previous example. Consider Figure \ref{2fgr:mainex5}.

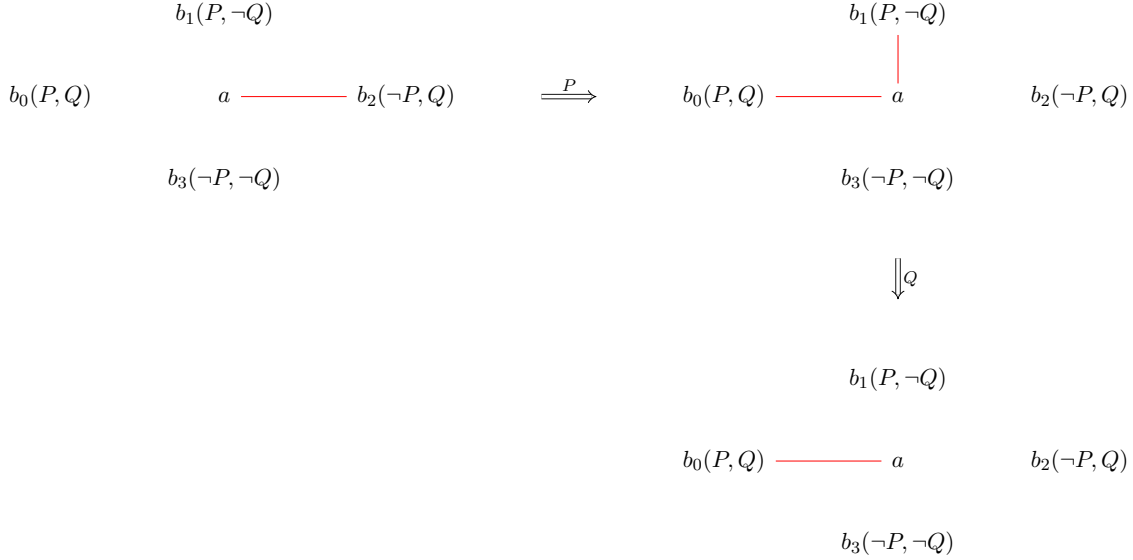
\begin{figure}[htbp]
	\begin{center}\adjustbox{scale=.8}{\begin{tikzcd}
				& b_1\color{black}(P,\neg Q)      &                                              &                               &    &                                         & b_1\color{black}(P,\neg Q)        &                                              \\
				b_0\color{black}(P,Q) & a \arrow[red, r, no head]                  & b_2\color{black}(\neg P,Q) & {} \arrow[r, "P", Rightarrow] & {} & b_0\color{black}(P,Q) & a \arrow[red, u, no head] \arrow[red, l, no head] & b_2\color{black}(\neg P,Q) \\
				& b_3\color{black}(\neg P,\neg Q) &                                              &                               &    &                                         & b_3\color{black}(\neg P,\neg Q)   &                                              \\
				&                                                  &                                              &                               &    &                                         & {} \arrow[d, "Q", Rightarrow]                      &                                              \\
				&                                                  &                                              &                               &    &                                         & {}                                                 &                                              \\
				&                                                  &                                              &                               &    &                                         & b_1\color{black}(P,\neg Q)        &                                              \\
				&                                                  &                                              &                               &    & b_0\color{black}(P,Q) & a \arrow[red, l, no head] & b_2\color{black}(\neg P,Q) \\
				&                                                  &                                              &                               &    &                                         & b_3\color{black}(\neg P,\neg Q)   &                                             
		\end{tikzcd}}
		\caption{$S_c$ before and after $P_a$ is publicly announced, assuming revision happens only when a perspective is no longer contained in any world.}\label{2fgr:mainex5}
		
	\end{center}
\end{figure}

Let us track how $S$ alters with each update in Figure \ref{2fgr:mainex5}. In the initial configuration, $\F(S)=\M(N,V,L)=\{\{a,b_0\},\{a,b_1\},\{a,b_2\},\{a,b_3\}\}$. After $P$ is announced, $\F(S_P)=\F(S)\cap\mathbf{F}_P=\{\{a,b_0\},\{a,b_1\}\}$. This does not affect the first update, except insofar as the worlds $\{a,b_2\}$ and $\{a,b_3\}$ are no longer available for subsequent revisions. This fact, however, $\F(S_{P,Q})=\F(S_P)\cap\mathbf{F}_Q=\{\{a,b_0\}\}$. This means the only available world to revise to is $\{a,b_0\}$, and indeed, this is what we see in Figure \ref{2fgr:mainex5}. The issue where the agent has gone back to considering the world $\{a,b_2\}$ possible, despite it being ruled out in the first update, has been taken care of.

Another potential solution to the issue of memory is to instead ``write'' what has been announced directly into perspectives. In the conception of memory given above, announcements are taken to be indefeasible. But this new conception would treat announcements as ``hard'' information available to the agents. Working this idea out formally does require altering definitions somewhat. In particular, we will need to allow for ``inconsistent'' perspectives, that is, perspectives that may contain both $P$ and $\neg P$ for some atom $P$. To account for this, Let $L:N\rightarrow 4^\mathfrak{P}$, where $L(n)(P)=4$ is interpreted as $P$ and $\neg P$ both being assigned to $n$. The maximal complex is restricted so that no faces containing a node for which some atom is assigned to $4$ are included, in order to preserve the consistency of every face.

This notion of update is in greater accordance with the kinds of updates explored in the literature connecting simplicial complexes and distributed computing. \cite{SimpDEL,DCTCT} There, announcements change, and generally increase, what is true at a perspective.\footnote{In the case where announcements are not public, the number of perspectives increases at each time step. Forthcoming work will seek to model this via action models.} For our purposes, for this notion of announcement to make sense, we must assume that only atoms are announced.\footnote{We leave the possibility of generalizing this idea for future work.} We define this as follows:

\begin{definition}{$\mathcal{M}[P]_h$}\label{2DefRevh}
	
	Consider a UCF simplicial model $\mathcal{M}=\langle N,V,L,\{S_a\}_{a\in A}\rangle$, and $\varphi$ a formula with no modalities in its recursive construction. Let  $N_P:=N$, $V_P:=V$, and 
	
	$$L_P(n)(P):=\begin{cases}
		1	& \text{if}\;L(n)(P)\in{1,3} \\
		4	& \text{if}\;L(n)(P)\in\{2,4\}
	\end{cases}$$
	
	Further, let $\mathbf{F}_P$ be the subset of $\mathcal{F}(\mathfrak{M}((N_P,V_P,L_P))$ which satisfy $P$. We need to define a function $\mathcal{R}_a:\mathcal{F}(S_a)\times\L_{B}(Ag)\rightarrow2^{\mathbf{F}_P}$.

	$$\mathcal{R}_a(F,P):=\{G\in\mathbf{F}_P~|~(\pi_a(F)=\pi_a(G))\wedge\forall X\in\mathbf{F}_P((\pi_a(X)=\pi_a(G))\rightarrow(|F\cap G|\geq|X\cap G|))\}$$
	
	Note that $\mathcal{R}_a:\F(S_a)\times\L_{B}(Ag)\rightarrow 2^{\mathbf{F}_P}$ is a nearness function. Further, note that, if $G\in\mathbf{F}_P$, then $\mathcal{R}_a(G,P)=\{G\}$. Let each $S_{a,P}$ be the simplicial complex whose facets are elements of the set $\bigcup_{F\in\mathcal{F}(S_a)}\mathcal{R}_a(F,P)$. Then $\mathcal{M}[P]:=\langle N_P,V_P,L_P,\{S_{a,P}\}_{a\in A}\rangle$.
	
\end{definition}

Note that Figure \ref{2fgr:mainex6} is the same as Figure \ref{2fgr:mainex5}. These two definitions do not always give the same output, however. Consider Figure \ref{2fgr:mainex7}. On the left is the update given by our procedure, and on the right is the update given by introducing the announced formula to every perspective. This suggests a subtle difference between the nature of hard information at a perspective, and the sense in which announcements are hard information, or indefeasible. As we see in Figure \ref{2fgr:mainex7}, an indefeasible announcement of $P$ rules out a facet where $P$ is false simply because no agent has hard information of $P$. By contrast, if we interpret announcements instead as giving agents hard information in their perspectives, the facet remains.\footnote{Above, we acknowledged that one can interpret a knowledge modality on $S$. It seems that, since when we use $S$ as memory, we force it to update with the usual, factive account, this is entirely conceptually consistent with interpreting a knowledge modality on $S$. Again, however, we do not explore that here, in order to keep the focus on revision.}

\begin{figure}[htbp]
	\begin{center}\adjustbox{scale=.8}{\begin{tikzcd}
				& b_1\color{black}(P,\neg Q)      &                                              &                               &    &                                         & b_1\color{black}(P,\neg Q)        &                                              \\
				b_0\color{black}(P,Q) & a \arrow[red, r, no head]                  & b_2\color{black}(\neg P,Q) & {} \arrow[r, "P", Rightarrow] & {} & b_0\color{black}(P,Q) & a \arrow[red, u, no head] \arrow[red, l, no head] & b_2\color{black}(\neg P,Q,P) \\
				& b_3\color{black}(\neg P,\neg Q) &                                              &                               &    &                                         & b_3\color{black}(\neg P,\neg Q,P)   &                                              \\
				&                                                  &                                              &                               &    &                                         & {} \arrow[d, "Q", Rightarrow]                      &                                              \\
				&                                                  &                                              &                               &    &                                         & {}                                                 &                                              \\
				&                                                  &                                              &                               &    &                                         & b_1\color{black}(P,\neg Q,Q)        &                                              \\
				&                                                  &                                              &                               &    & b_0\color{black}(P,Q) & a \arrow[red, l, no head] & b_2\color{black}(\neg P,Q,P) \\
				&                                                  &                                              &                               &    &                                         & b_3\color{black}(\neg P,\neg Q,P,Q)   &                                             
		\end{tikzcd}}
		\caption{Figure \ref{2fgr:mainex5} using Definition \ref{2DefRevh} instead of \ref{2DefRevm}.}\label{2fgr:mainex6}
		
	\end{center}
\end{figure}
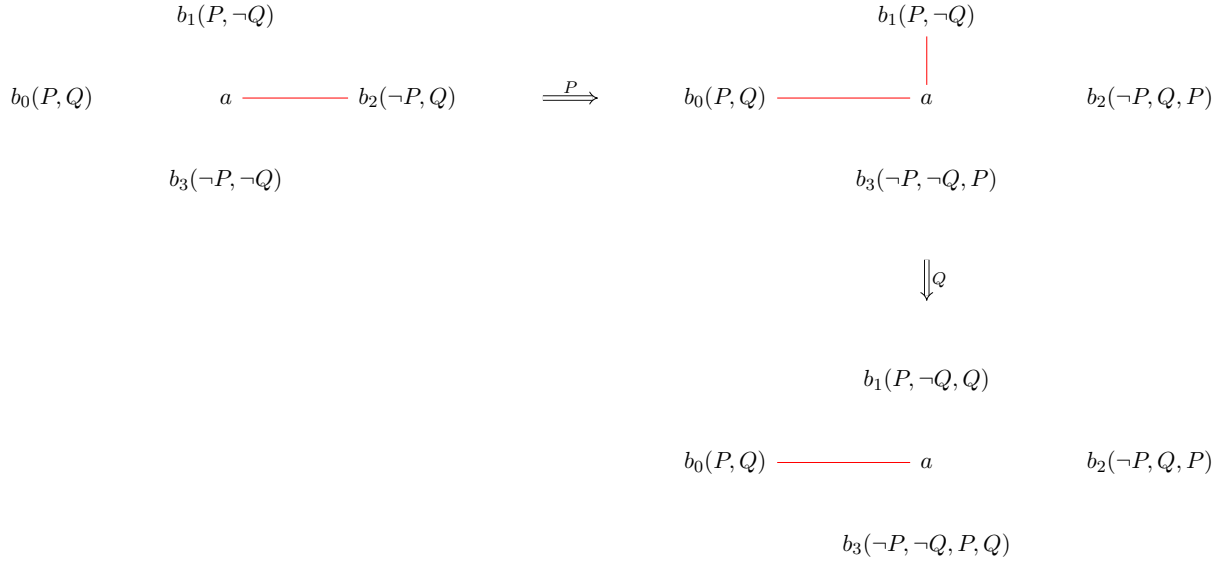

\begin{figure}[htbp]
	\begin{center}
		\begin{tikzcd}
			\color{red}a \arrow[d, no head] &  & \color{red}a \arrow[d, no head]                 \\
			\color{blue}b                   &  & \color{blue}b                                   \\
			{} \arrow[d, "P"', Rightarrow]  &  & {} \arrow[d, "P", Rightarrow]                   \\
			{}                              &  & {}                                              \\
			\color{red}a                    &  & \color{red}a\color{black}(P) \arrow[d, no head] \\
			\color{blue}b                   &  & \color{blue}b\color{black}(P)                  
		\end{tikzcd}
		\caption{Two sequences of two models. The first, on the left, shows an initial configuration with two perspectives composing one world. After $P$ is announced, using Definition \ref{2DefRevm}, that world is erased. On the right, the same initial configuration is given and the same announcement is made, but using Definition \ref{2DefRevh}, the world remains.}\label{2fgr:mainex7}
		
	\end{center}
\end{figure}
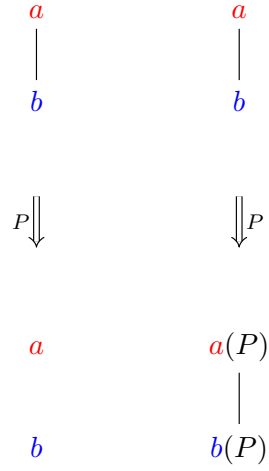

It is worth emphasizing that this notion of ``memory'' is only well motivated if announcements are taken to be, in some sense, indefeasible. If announcements can be contradicted, Figure \ref{2fgr:mainex4} might be a coherent picture. This is because this notion of ``memory'' would ensure that if an agent learns something that entails $\varphi$, and then later learns something which entails $\neg\varphi$, there are no worlds for the agent to revise to. Future work will want to consider precisely such kinds of contradictory announcements, and so will not make use of this notion of memory.

\section{Conclusion} \label{2sec:conc}

This paper sought to give a model of belief revision in simplicial semantics. In particular, the hope was that by appealing to the fact that in simplicial semantics, worlds are composed of a different structure, namely perspectives, we could say that two worlds were ``nearer'' if they shared more perspectives. Two different formalisms were given on the basis of this intuition, the first of which allowed revision on every world ruled out by an announcement, and the second of which restricted revision to situations where perspectives become ``isolated'', that is, an $a$-colored perspective is no longer contained in a facet of $S_a$. After exploring some examples, we saw that, if we took announcements to be indefeasible, repeated revisions led to some strange behavior, where agents could go back-and-forth on whether or not they considered a particular world possible, simply believing whatever was announced to them most recently. To solve this issue, we introduced the notion of ``memory'', and motivated two different ways we could force agents to ``remember'' what had been announced, avoiding the previous issue.

There are quite a few directions for future work. A simple project would involve connecting the use of $S$ as ``memory'' and the possibility of interpreting a knowledge modality on $S$, as in \cite{BelSimp}. Whether or not these two formal applications of $S$ are conceptually consistent with each other will take further analysis. Another possibility for future work is exploring how different notions of trust affect this model of revision. Indeed, if some agents treat others as ``more trustworthy'' than others, presumably their perspectives should ``outrank'' other agents' perspectives in determining nearness. Spelling this out formally, and seeing all the possible ways one might do this, would be a compelling direction to take this notion of revision.

It would be nice to explore more connections between this semantics and other presentations of belief revision. In particular, future work should explore connections between our semantics and the ``AGM'' presentation of revision. \cite{AGM,Grove} For example, we should check if and how we can use our semantics to model multi-agent belief revision. \cite{MAGM} Another possibility would be to use the same perspective-based structure of worlds in simplicial semantics to model belief merging instead of belief revision. We can therefore phrase the question of merging thusly: Given a model, and two perspectives in the model, what is the ``best'' set of worlds one could create using the two sets of worlds containing each of the two perspectives, respectively. A simple answer is obviously the intersection. Given two perspectives, call them $x_1$ and $x_2$, let $S_i$ be the set of facets in a model $\mathcal{M}$ containing $x_i$. Then merging the perspectives $x_1$ and $x_2$ could be given by the set of worlds $S_1\cap S_2$. But we could obviously say something more nuanced. Perhaps $X\in S_1$ but $X\notin S_2$. However, there is a world $X'$ that differs from $X$ in only one perspective, and $X'\in S_1\cap S_2$. Perhaps the merge set should include worlds like $X'$. Similar to our discussion of announcements, one can also modify this so that the update is $S_1\cap S_2$, unless this intersection is empty, in which case one selects some set of nearby worlds to $S_1\cup S_2$.

The most fleshed out project for future work involves using simplicial semantics to model Gossip Protocols, and to incorporate revision into those protocols. \cite{Birman07} The idea would be to use action models to model protocols as a first step. This already has had some exploration in the literature, including the simplicial semantics literature. \cite{SimpDEL,KaSC,KaG,ActionDC,CPM,CPaAM} Once a definition of action model is sufficiently worked out for our purposes, the second step would be to modify this definition of action model to update via revision. This, too, already has some history of work. \cite{BRDC1,BRDC2} With all of these pieces in play, the goal would be to then describe protocols where the update procedure is interpreted as revision. A possible future step after this would be to simulate said protocols, with the potential goal of comparing protocols with and without revision in order to assess the efficacy of revision as a learning mechanism in certain contexts.

\bibliographystyle{plain}	
\bibliography{Epistemology} 

\section{Proofs} \label{2app:prf}

This section largely sets out to prove a completeness theorem. For this theorem, along the way we show a categorical equivalence between a category of Kripke models and a category of simplicial models. Only one direction of this equivalence is needed for the proof of completeness. However, showing the full categorical equivalence gives a sense of what kind of ``information'' is preserved when translating back and forth between the two settings. It's also consistent with how such proofs are done in the literature. \cite{KaSC,SimpDEL}

\begin{proof}[Proof of Theorem \ref{2Thm:Comp}]\label{2prf:comp}
	The proof follows many of the standard techniques in the literature. \cite{KaSC}
	
	\begin{definition}[$\mathfrak{S}$]\label{2DefCS} Let $S_1$ and $S_2$ be \textit{UCF} simplicial models. We say that $f:S_1\rightarrow S_2$ is a morphism iff, treating $f$ as a function on $\mathcal{F}(S_1)$, we have that for any facets $X,Y\in S_1$ and $X\in S_{1,b}$, $Y\in S_{1,a}$ and $\pi_a(X)=\pi_a(Y)$ iff $f(Y)\in S_{2,a}$ and $\pi_a(f(X))=\pi_a(f(Y))$, and furthermore, $\bigcup_{a\in A}L(\pi_a(X))=\bigcup_{a\in A}L(\pi_a(f(X)))$. We say that the category whose objects are \textit{UCF} simplicial models with these morphisms is $\mathfrak{S}$.
	\end{definition}
	
	\begin{definition}[$\mathfrak{K}$]\label{2DefCK} Let $K_1$ and $K_2$ be transitive and Euclidean proper Kripke models such that for any world $w$ and atom $P$ true at that world, there exists an agent $a\in A$ such that if $wR_au$ and $w\in \ell(P)$, then $u\in \ell(P)$. It's easy to see that this makes \textbf{NU} sound. We call such Kripke models \textit{Ideal}. We say that $g:K_1\rightarrow K_2$ is a morphism iff, treating $g$ as a function on worlds, $wR_{a,1}u$ iff $g(w)R_{a,2}g(u)$, and $w\in \ell_1(P)$ iff $g(w)\in \ell_2(P)$. We say that the category whose objects are ideal Kripke models with these morphisms is $\mathfrak{K}$.
	\end{definition}
	
	\begin{lemma}\label{2ThmESK} $\mathfrak{S}$ and $\mathfrak{K}$ are equivalent categories.
	\end{lemma}
	
	\begin{proof}
		First we need to define two functors. For the first functor, define $F:\mathfrak{S}\rightarrow\mathfrak{K}$ such that for a \textit{UCF} model $S$, $F(S)$ is a Kripke model such that for each maximal facet $X$ in some $S_a$, there is a unique world $w_X$ in $F(S)$, and $w_X R_a w_Y$ iff $\pi_a(X)=\pi_a(Y)$ and $Y\in\mathcal{F}(S_a)$. By design, worlds and facets are in bijection. Furthermore, $w_X\in\ell(P)$ iff $P\in\bigcup_{x\in X}L(x)$. We need to show that this is an ideal Kripke model. It's obvious that $R_a$ is transitive and Euclidean for each $a\in A$. Now consider any atom $P$ and world $w_X$. Suppose that $w_X\in\ell(P)$. It follows that $P\in\bigcup_{x\in X}L(x)$. Fix $x\in X$ such that $P\in L(x)$. Since $S$ is $\textit{UCF}$, fix $a$ such that $V(x)=a$. Suppose further that $w_X R_a w_Y$. Then $x=\pi_a(X)=\pi_a(Y)$, and so $P\in L(\pi_a(Y))\subseteq\bigcup_{x\in Y}L(Y)$. It follows that $w_Y\in\ell(P)$, as desired. Finally, if $w_X R_a w_Y$ for all $a\in A$, then $\pi_a(X)=\pi_a(Y)$ for all $a\in A$. This means that $X=Y$, and so $w_X=w_Y$, which tells us that $F(S)$ is a proper Kripke model. This shows that we have an ideal Kripke model.
		
		Now we must show that $F$ as given above is a functor. Let $f:S_1\rightarrow S_2$ be a $\mathfrak{S}$ morphism. Then $F(f):F(S_1)\rightarrow F(S_2)$ is defined as follows: $F(f)(w_X)=w_{f(X)}$. We need to show that this is a $\mathfrak{K}$ morphism. Indeed, by design, $w_{f(X))}$ is a world in $F(S_2)$. Let $R_1$ be the relation in $F(S_1)$ and $R_2$ be the relation in $F(S_2)$. Suppose $w_X R_1^a w_Y$. This is true if and only if $\pi_a(X)=\pi_a(Y)$ and $Y\in\mathcal{F}(S_{1,a})$. Call this shared vertex $v$. This follows if and only if $f(v)\in f(X)\cap f(Y)$, and since $f(Y)$ is in $S_{2,a}$, this is true if and only if $w_{f(X)}R_2^aw_{f(Y)}$, as desired. Suppose $w_X\in\ell_1(P)$. This is true iff $P\in\bigcup_{x\in X}L(x)$. Because $\bigcup_{x\in X}L(x)=\bigcup_{x\in f(X)}L(x)$, this is true if and only if $P\in\bigcup_{x\in f(X)}L(x)$. This is true iff $w_{f(X)}\in\ell_2(P)$, as desired. This shows that $F$ is well defined.
		
		It's obvious that $F$ preserves the identity. We must show that it preserves composition to establish its functoriality. Let $f_1:S_1\rightarrow S_2$ and $f_2:S_2\rightarrow S_3$ be morphisms. Then the following holds:
		
		$$F(f_2\circ f_1)(w_X)=w_{f_2\circ f_1(X)}=w_{f_2(f_1(X))}=F(f_2)(w_{f_1(X)})=F(f_2)(F(f_1)(w_X))=\left (F(f_2)\circ F(f_1)\right )(w_X)$$
		
		From this it follows that $F(f_2\circ f_1)=F(f_2)\circ F(f_1)$, as desired. So, we have given our first functor, and shown that it is a functor.
		
		Now we must define the second functor. Define $G:\mathfrak{K}\rightarrow\mathfrak{S}$ such that for an ideal model $K$, $G(K)$ is a simplicial model where the facets of $S$ are of the form $X_w$ for each world $w$ in $K$, and the facets of each $S_a$ are of the form $X_w$ for each world $w$ in $K$ such that $w R_a w$. Because $K$ is ideal, worlds and facets are in bijection. We assume each facet has a unique vertex for each agent. The only further stipulation is that $\pi_a(X_w)=\pi_a(X_u)$ and $X_u\in\mathcal{F}(S_a)$ if and only if  $w R_a u$.\footnote{By the Euclidean axiom, we have that, if $wR_a u$, then $uR_a u$, so $X_u\in\mathcal{F}(S_a)$ is actually redundant as a condition.} For each atom $P$ true at $w$, we say that $P\in L(\pi_a(X_w))$ iff $P$ is true at all $u$ such that $wR_au$. By the fact that $K$ is ideal, we know such an $a$ exists. It is easy to see that this is a \textit{UCF} model. Indeed, the disjoint collection of facets $X_w$ is vacuously \textit{UCF}, and each association of two vertices preserves the \textit{UCF} property.
		
		Now we must show that $G$ as given above is a functor. Let $g:K_1\rightarrow K_2$ be a $\mathfrak{K}$ morphism. Then $G(g):G(K_1)\rightarrow G(K_2)$ is defined as follows: $G(g)(X_w)=X_{g(w)}$. We need to show this is a $\mathfrak{S}$ morphism. It suffices to show that $G$ preserves facets and the logical information. Since $g$ preserves reflexive edges, $G(g)$ therefore preserves facets in each $S_{1,a}$ in $G(K_1)$. Let $X$ be a facet in $S_1$ in $G(K_1)$. Fix $w$ in the worlds of $K_1$ such that $X=X_w$. Then $G(g)(X_w)=X_{g(w)}$ which is by definition a facet in $G(K_2)$, since $g(w)$ is a world in $K_2$. Now suppose $P\in L(\pi_a(X_w))$. This is true iff $u\in\ell(P)$ for all $u$ such that $w R_a u$. This is true iff $g(u)\in\ell(P)$ for all $u$ such that $w R_a g(u)$. This is true iff $P\in L(\pi_a(g(w)))$. This is true iff $P\in L(G(g)(\pi_a(X_w)))$, as desired. This shows that $G$ is well defined.
		
		It's obvious that $G$ preserves the identity. We must show that it preserves composition to establish its functoriality. Let $g_1:K_1\rightarrow K_2$ and $g_2:K_2\rightarrow K_3$ be morphisms. Then the following holds:
		
		$$G(g_2\circ g_1)(X_w)=X_{g_2\circ g_1(w)}=X_{g_2(g_1(w))}=G(g_2)(X_{g_1(w)})=G(g_2)(G(g_1)(X_w))=\left (G(g_2)\circ G(g_1)\right )(X_w)$$
		
		From this it follows that $G(g_2\circ g_1)=G(g_2)\circ G(g_1)$, as desired.
		
		What remains to be shown is that for any Kripke model $K$, $FG(K)$ is isomorphic to $K$ and for any Simplicial model $S$, $GF(S)$ is isomorphic to $S$. For the first case, note that the worlds in $FG(K)$ are of the form $w_{X_w}$ where $w$ is a world in $K$. The above demonstrated this identification is a bijection. Moreover, for any worlds $w$ and $u$ in $K$, $wR_au$ in $K$ iff $X_u$ is in $S_a$ and $\pi_a(X_w)=\pi_a(X_u)$. This is true iff $w_{X_w}R_au_{X_u}$, as desired. Moreover, $w\in\ell(P)$, iff there is some $a\in Ag$ such that $P\in L(\pi_a(X_w))$, which is true iff $w_{X_w}\in\ell(P)$. This shows that $K$ and $FG(K)$ are isomorphic.
		
		Suppose $X$ and $Y$ are facets in $S$, $\pi_a(X)=\pi_a(Y)$ and $Y\in S_a$. This is true iff $w_X R_a w_Y$, which, because $Y\in S_a$, and so $w_Y R_a w_Y$, is true iff $\pi_a(X_{w_X})=\pi_a(Y_{w_Y})$ and $Y_{w_Y}\in\mathcal{F}(S_a)$ for $GF(S)$. Suppose instead that $\bigcup_{a\in A}L(\pi_a(X))=\bigcup_{a\in A}L(\pi_a(Y))$ This is true iff $w_X,w_Y\in\ell(P)$. Because $F(S)$ is ideal, for each $P$, there is an agent $a\in Ag$ such that if $w\in\ell(P)$, then for all $u$ such that $wR_au$, $u\in\ell(P)$. By definition of $G$, this means that $\bigcup_{a\in A}L(\pi_a(X_{w_X}))=\bigcup_{a\in A}L(\pi_a(Y_{w_Y}))$. This shows that $S$ and $GF(S)$ are isomorphic.
	\end{proof}
	
	Now, we need to show that $F$ and $G$ are logic preserving functors. This is a simple structural induction.
	
	Suppose $P$ is an atomic formula, $K$ is a Kripke model, and $S$ is a simplicial model. Suppose further that $K,w\vDash P$ and $S,X\vDash P$. Then, $w\in\ell(P)$ and there is $a\in Ag$ such that $P\in L(\pi_a(X))$. So, $w_X\in\ell(P)$, and $P\in L(\pi_a(X_w))$, showing that $F$ and $G$ preserve atomic formulas.
	
	Assume the obvious inductive hypothesis for formulas up to depth $n$. The only interesting case is the modal one. Let $\varphi$ be a formula of depth $n$. Suppose that $K,w\vDash B_a\varphi$ and $S,X\vDash B_a\varphi$. We will handle the Kripke case first. This means that, for all $u$ such that $wR_au$, $K,u\vDash\varphi$. By the IH, and the definition of $G$, for all facets $Y\in S_a$ such that $\pi_a(Y)=\pi_a(X_u)$, $G(K),X_u\vDash\varphi$. The result follows. Because $S,X\vDash B_a\varphi$, then for all $Y\in S_a$ such that $\pi_a(X)=\pi_a(Y)$, then $S,Y\vDash\varphi$. By the IH, and the definition of $F$, $F(S),w_Y\vDash\varphi$. The result follows.
	
	The next step is to construct a canonical \textbf{K45+NU} frame model, using the usual ``unboxing'' method. That is, construct a model $K_C$ whose set of worlds $W_C$ is the set of all maximal consistent \textbf{K45+NU} sets of formulas, $\ell(P)=\{w\in W_C|P\in w\}$, and $wR_au$ if and only if all formulas $\varphi$, if $B_a\varphi\in w$, then $\varphi\in u$. It is easy to show that this model is such that $K_C,w\vDash\varphi$ if and only if $\varphi\in w$. In general, this model is not proper, but we can apply the translation from \cite{Note1} to generate a bisimilar model which is. Then, applying $G$ to this proper model gives us a canonical simplicial model. Completeness follows.
	
\end{proof}
	
\end{document}